\documentclass[article]{imsart}

\usepackage[utf8]{inputenc}      
\usepackage[T1]{fontenc}           
\usepackage[english]{babel}

\usepackage{booktabs}
\usepackage{amssymb}               
\usepackage{amsmath}               
\usepackage{amsthm}
\usepackage{subfigure}
\usepackage{graphicx}
\graphicspath{{figures/}}
\usepackage{hyperref}
\usepackage{algorithm}
\usepackage{algpseudocode}

\usepackage{changes}

\usepackage{todonotes}

 \renewcommand{\added}[1]{#1}

\newtheorem {Proposition}{Proposition}[section]
\newtheorem {Lemma}[Proposition] {Lemma}
\newtheorem {Theorem}[Proposition]{Theorem}
\newtheorem {Corollary}[Proposition]{Corollary}
\newtheorem {Remark}[Proposition]{Remark}

 {
     \theoremstyle{plain}
     \newtheorem{assumption}{Assumption}
 }

\def\N{\mathbb{N}}
\def\Z{\mathbb{Z}}

\def\R{\mathbb{R}}

\def\x{\mathbf{x}}
\def\z{\mathbf{z}}
\def\y{\mathbf{y}}
\def\p{\mathbf{p}}

\def\bx{\bar{\mathbf{x}}}
\def\by{\bar{\mathbf{y}}}
\def\bz{\bar{\mathbf{z}}}
\numberwithin{equation}{section}

\startlocaldefs
\endlocaldefs

\begin{document}

\begin{frontmatter}

\title{Two-sample goodness-of-fit tests on the flat torus based on Wasserstein distance  and their relevance to structural biology}
\runtitle{???}

\author{\fnms{Javier} \snm{González-Delgado}$^{1,2}$\ead[label=e1]{javier.gonzalez-delgado@math.univ-toulouse.fr}},
\author{\fnms{Alberto} \snm{González-Sanz}$^{1,3}$\ead[label=e2]{alberto.gonzalez$\_$sanz@math.univ-toulouse.fr}},
\author{\fnms{\\Juan} \snm{Cortés}$^{2}$\ead[label=e4]{juan.cortes@laas.fr}}
\and
\author{\fnms{Pierre} \snm{Neuvial}$^{1}$\ead[label=e3]{pierre.neuvial@math.univ-toulouse.fr}}

\address{\printead{e1}}
\address{\printead{e2}}
\address{\printead{e4}}
\address{\printead{e3}}

\address{$^{1}$ Institut de Mathématiques de Toulouse, UMR 5219, Université de Toulouse, CNRS.}
\address{$^{2}$ LAAS-CNRS, Université de Toulouse, CNRS.}
\address{$^{3}$ ImUva, Universidad de Valladolid.}

\runauthor{González-Delgado et al.}
\runtitle{Wasserstein tests on the flat torus}

\begin{abstract}
This work is motivated by the study of local protein structure, which is defined by two variable dihedral angles that take values from probability distributions on the flat torus. Our goal is to provide the space $\mathcal{P}(\mathbb{R}^2/\mathbb{Z}^2)$ with a metric that quantifies local structural modifications due to changes in the protein sequence, and to define associated two-sample goodness-of-fit testing approaches. Due to its adaptability to the geometry of the underlying space, we focus on the Wasserstein distance as a metric between distributions. 

We extend existing results of the theory of Optimal Transport to the $d$-dimensional flat torus $\mathbb{T}^d=\mathbb{R}^d/\mathbb{Z}^d$, in particular a Central Limit Theorem for the fluctuations of the empirical optimal transport cost. Moreover, we propose different approaches for two-sample goodness-of-fit testing for the one and two-dimensional case, based on the Wasserstein distance. We prove their validity and consistency. We provide an implementation of these tests in \textsf{R}. Their performance is assessed by numerical experiments on synthetic data and illustrated by an application to protein structure data. 
\end{abstract}





\end{frontmatter}

\tableofcontents


\section{Introduction}

When it comes to measure the distance between two probability distributions, the well known Wasserstein distance, derived from the theory of Optimal Transport (OT),
provides both strong theoretical guarantees --it metrizes weak convergence \cite{Villani2008}-- and attractive empirical performance \cite{Cuturi2018}.
Most of the applications of such theory are related to the very active field of machine learning, notably in the framework of generative networks \cite{arjovsky}, robustness \cite{serrurier2020achieving} or fairness \cite{Gordaliza2019}, among others.

From a statistical point of view, one of the main caveats of the theory of OT comes from the \emph{curse of dimensionality}: the rate of convergence of the empirical Wasserstein distance decreases as $n^{-1/d}$ with the dimension \cite{Fournier2016}.
Another important issue is the asymptotic behavior of the fluctuations of the empirical optimal transport cost
. For probability measures supported in $\R^d$, it has been proved, using Efron–Stein's inequality that, for the cost $L^2$, the difference $\sqrt{n}(\mathcal{W}_2^2(P_n,Q)-\mathbb{E}\mathcal{W}_2^2(P_n,Q))$ is asymptotically Gaussian \cite{DelBarrio2019}. Recently, the proofs have been extended to some general costs in $\R^d$, including the cost $L^p$, for $p>1$ \cite{Alberto}. Concerning statistical goodness-of-fit tests based on Wasserstein distance, the one-sample case has already been addressed in \cite{Hallin2021} and, when the probability distributions are defined over $\mathbb{R}$, two-sample tests can be derived from \cite{Barrio1999, Munk1998NonparametricVO}.

In this paper, we focus on the $d$-dimensional flat torus $\mathbb{T}^d:=\mathbb{R}^d/\mathbb{Z}^d$ where, even from the purely theoretical point of view, OT has not been completely addressed, besides the work in \cite{CORDEROERAUSQUIN1999199}, \cite{McCann} or, more recently, in \cite{manole2021plugin}.  However, this space appears naturally when the probability measures are periodic (e.g. for distributions of angles). The main objective of this work is (1) to extend recent existing OT results to the space of probability measures on the flat torus $\mathcal{P}(\mathbb{T}^d)$, especially a Central Limit Theorem (CLT) for the fluctuations of the empirical optimal transport cost, and (2) to address in particular the two-dimensional case, by constructing two-sample goodness-of-fit tests based on the Wasserstein distance.

Our motivation for extending the theory of OT to $\mathbb{T}^2$ comes from the investigation of proteins. Understanding the relationships between protein sequence, structure and function is the main goal of Structural Biology. In addition to its scientific importance, a better understanding of these relationships is essential for applications in diverse areas, such as biomedicine and biotechnology. The conformational state of a protein can be defined by a vector of angles, corresponding to rotations around the chemical bonds between the atoms that constitute its ``backbone''. This vector contains two values per amino-acid, $\phi$ and $\psi$, which follow a certain distribution, and which are usually represented using the so-called Ramachandran plots \cite{Ramachandran:1963} (see also Figure~\ref{fig:example_flory_ejs}). The analysis of these distributions has several important applications, such as the validation or refinement of protein structures determined from biophysical techniques \cite{Morris:1992,Lovell:2003}, the prediction of some biophysical measurements to complement experiments \cite{Shen2017}, and the development of potential energy models or scoring methods for protein structure modeling, prediction and design \cite{Betancourt2004,Rata2010,Ting2010}.

In this context, the definition of a suitable distance between distributions on $\mathbb{T}^2$ is essential. This would allow to quantify the expected magnitude of structural effects associated with local changes in the sequence, and therefore  to develop improved versions of the aforementioned modeling and prediction techniques. Nevertheless, this has not been done satisfactorily in previous works. For example, significant differences between two laws are stated after visual comparison of two empirical distributions in \cite{Rata2010} and \cite{Shen2017}, and the Hellinger distance is used to compare distributions on a non-periodic $[-\pi,\pi]\times[-\pi,\pi]$  in \cite{Ting2010}. Powerful statistical tests remain to be defined and implemented in order to state such differences, being based on a metric that takes geometry into consideration. As many other commonly-used metrics, Hellinger distance ignores the underlying geometry of the space. Here, we propose to use the Wasserstein distance, whose advantageous geometrical and mathematical properties are described in \cite{Cuturi2018}, \cite{Villani2003} and \cite{Villani2008}, to define goodness-of-fit testing techniques for two measures on $\mathbb{T}^2$, allowing a more accurate study of the distribution of protein local conformations.\\

The paper is organized as follows:
\begin{itemize}
\item Section~\ref{section_measures} starts by introducing  the general framework of measures on the flat torus in general dimension,  followed by the precise formulation of the optimal transport problem. Section~\ref{sectionolutions} is devoted to the study of the shape of the solutions, recalling that they are the gradients of periodic convex functions and showing the uniqueness of the potential in Corollary~\ref{Corollary:equals}.  Section~\ref{section_asyimptotics} proves through Theorem~\ref{Theoremeq} that the optimal transport potentials converge, up to an additive constant, when the measures converge weakly. This result implies that the method of \cite{DelBarrio2019} based on Efron–Stein's inequality can be applied to derive a Central Limit Theorem,  see Theorem~\ref{Theo:centrallim2} in Subsection~\ref{section:normality}. Finally, we show how the previously defined CLT does not allow the definition of an asymptotic test.
\item Section~\ref{section:test} shows how Wasserstein distance can be used to define two-sample goodness-of-fit tests in the two-dimensional flat torus. We propose two testing approaches. The first one, introduced in Subsection \ref{sec:test_geodesics}, consists in testing the equality of two measures projected into a finite number of closed geodesics on $\mathbb{T}^2$. The second, presented in Subsection \ref{sec:upper_bound}, is a conservative procedure based on upper-bounding the exact $p$-values. This is possible thanks to a concentration inequality given in Theorem~\ref{Theorem::bound}, together with faster convergence rates for the expectation.
\item Section~\ref{simulations} reports numerical experiments illustrating the relevance of these theoretical results, first with synthetic data and then with real data from protein structures, showing that our methods behave well in both cases.
\end{itemize}
To facilitate reading, the proofs are relegated to the Appendix, but in some cases the intuition behind the proof is provided in the main text for clarity.

\section{Optimal transport in \texorpdfstring{$\mathbb{R}^d/\mathbb{Z}^d$}{the flat torus}}\label{section_measures}\label{first_section}

Let $\mathbb{T}^d:=\mathbb{R}^d/\mathbb{Z}^d$ be defined as the quotient space derived from the equivalence relation $\x \mathcal{R} \y $ if $\x-\y\in \Z^d$. For each $\x\in \R^d$ we denote as $\bx\in \mathbb{T}^d$ its equivalence class and reserve the notation $\tau$ for the canonical projection map
$
\x\mapsto \tau(\x)=\bx.
$
The topology of the quotient space is defined as the finest one that makes $\tau$ continuous.  With this topology, the space $\mathbb{T}^d$ is a Polish space with the distance derived from the Euclidean norm $\|\cdot \|$,
\begin{equation}\label{ground_distance}
d(\bx,\by):=\inf_{\p\in\Z^d}\| \x-\y+\p\|.
\end{equation}
Note that the last claim is true since the projection map $\tau$ is in fact a metric identification,  $(\R^d,\|\cdot \|)$ is a Banach space and $\Z^d$ is a closed subset, then it is complete, metrizable through $d$ and separable.  

Set $p> 1$. For two probability measures $P,Q\in \mathcal{P}(\mathbb{T}^d)$,  a probability measure $\pi \in \mathcal{P}(\mathbb{T}^d\times\mathbb{T}^d)$ is said to be an \emph{optimal transport plan for the cost $d^p$} between $P$ and $Q$ if it solves
\begin{align}\label{kant}
\mathcal{T}_p(P,Q) := \inf_{\gamma \in \Pi(P,Q)}\int_{\mathbb{T}^d\times \mathbb{T}^d} d^p(   \bx,\by) d \gamma(\bx, \by),
\end{align}
where $\Pi(P,Q)$ is the set of probability measures $\gamma\in \mathcal{P}(\mathbb{T}^d \times\mathbb{T}^d)$  such that $\gamma(A\times \R^d)=P(A)$ and $\gamma(\mathbb{T}^d \times B)=Q(B)$ for all Borel measurable subsets $A,B$ of $\mathbb{T}^d$. 

The Kantorovich problem \eqref{kant} can be formulated in a dual form, as follows
\begin{align}\label{dual}
\mathcal{T}_p(P,Q)=\sup_{(f,g)\in \Phi_p(P,Q)}\int_{\mathbb{T}^d} f(\bx) dP(\bx)+\int_{\mathbb{T}^d} g(\by) dQ(\by),
\end{align}
where $$\Phi_p(P,Q)=\{ (f,g)\in L_1(P)\times L_1(Q): \ f(\bx)+g(\by)\leq d^p(\bx,\by)\quad\forall\, \bx,\by\in\mathbb{T}^d\}.$$
The element $\psi\in L_1(P)$ is said to be an \emph{optimal transport potential from $P$ to $Q$ for the cost $d^p$} if there exists $\varphi\in L_1(Q)$ such that the pair $(\psi, \varphi)$ solves \eqref{dual}.  Recall from \cite{Villani2008} that the solutions of \eqref{dual} are pairs $(f,g)$ of $d^p$-conjugate $d^p$-concave functions.  This means  that
\begin{align}\label{dconcave}
f(\x)=\inf_{\by\in\mathbb{T}^d}\{d(\bx,\by)^p-g(\by)\} \ \text{and} \ \ g(\by)=f^{d^p}(\by)=\inf_{\bx\in\mathbb{T}^d}\{d(\bx,\by)^p-f(\bx)\} .
\end{align}
Furthermore, since $\mathbb{T}^d$ is a Polish space, then Theorem 4.1  in  \cite{Villani2008} implies that there exists a solution $\pi^*$ of \eqref{kant}. Additionally, Theorem 5.10 in \cite{Villani2008} establishes that  $\text{supp}(\pi^*)$ is \emph{$d^p$-cyclically monotone}. This means that for any finite sequence $\{(\x_k,\y_k) \}_{k=1}^n\subset \text{supp}(\pi^*)$ and any bijection ${\sigma:\{ 1, \dots, n \}\rightarrow \{ 1, \dots, n \}}$, the following inequality holds:
\begin{equation*}
\sum_{k=1}^n d^p(\bx_k, \by_k)\leq \sum_{k=1}^n d^p(\bx_k, \by_{\sigma(k)}).
\end{equation*}
Note that, if $Q$ is a probability measure in $\mathbb{T}^d$, its support is defined as the closed set $\text{supp}(Q)\subset\mathbb{T}^d$ composed by $\bx\in\mathbb{T}^d$ such that for any neighborhood $\mathcal{U}_{\bx}$ of $\bx$ it holds that $Q(\mathcal{U}_{\bx})>0$. The interior of the support is denoted by $\mathcal{X}_Q.$ 

With the same obvious notation we can define a \emph{$\| \cdot\|^p$-cyclically monotone} set. Note that for $p=2$,  {$\| \cdot\|^2$-cyclical monotonicity } is equivalent to the concept of cyclical monotonicity in convex analysis,  described in  \cite{RockConvex}. Recall that a set $A\subset\R^d\times \R^d$ is \emph{cyclically monotone} if for every finite sequence $\{(\x_k,\y_k) \}_{k=1}^n\subset A$ and every bijection ${\sigma:\{ 1, \dots, n \}\rightarrow \{ 1, \dots, n \}}$ it holds that
\begin{equation*}
\sum_{k=1}^n \langle \x_k, \y_k\rangle\geq \sum_{k=1}^n \langle \x_k, \y_{\sigma(k)}\rangle.
\end{equation*}
Consequently, the concept of $d^p$ (resp. $\|\cdot\|^p$) -cyclical monotonicity is the natural generalization, to other spaces and costs, of cyclical monotonicity. 

In some cases, that we will study later on,  there exists some measurable map $T$ such that the optimal transport plan $\pi$ satisfies $\pi=\left(I\times T\right){\#}P$, where the symbol $T{\#}P$ denotes the \emph{push forward measure} of $P$ through $T$, which is defined by $T{\#}P(A):=P(T^{-1}(A))$, for all measurable $A\subset\mathbb{T}^d$, and $I$ denotes the identity map.
Therefore, the problem becomes equivalent to the following \emph{Monge} formulation:
\begin{align}\label{monge}
\mathcal{T}_p(P,Q)=\inf_{T{\#}P=Q}\int_{\mathbb{T}^d} d^p(   \bx,T(\bx)) d P(\bx).
\end{align}

\subsection{Existence of $\| \cdot \|^p$-cyclically monotone mappings.}\label{sectionolutions}
A cyclically monotone map is the natural generalization of a non decreasing function in the real line (as being the gradient of a convex function, see \cite{RockConvex}). Cyclical monotonicity provides a powerful tool for statistical studies, see \cite{Hallin2021, DELBARRIO2020104671, hallinChernozhukov2017} among others. The existence of cyclically monotone maps between probability measures in $\R^d$ has been investigated, in parallel, by  \cite{cuesta1989notes} and \cite{brenier1991polar}, with the restrictive assumption of finite second order moment, relaxed in \cite{McCann1995ExistenceAU}. For periodic measures, the celebrated result of \cite{CORDEROERAUSQUIN1999199} showed the existence. The concept of cyclically monotone map also appears naturally when solving an optimal transport problem with quadratic cost in $\R^d$. Therefore, for any potential cost $\| \cdot \|^p$, the natural generalization is the one of $\| \cdot \|^p$-cyclically monotone. In fact, \cite{GangboMccann} proved the existence of a $\| \cdot \|^p$-cyclically monotone mapping between  probability measures with finite moment of order $p>1$. To the authors' knowledge, no previous work has dealt with the existence of $\| \cdot \|^p$-cyclically monotone mappings between periodic probability measures. Consequently, the main result of this section is Theorem~\ref{TheoremExistenceOfArrangements}, which  shows the existence and uniqueness of a $\| \cdot\|^p$-cyclically monotone preserving map $\mathbf{S}_p$ between periodic measures, for $p>1$, and relates it with the solution of \eqref{monge}.  Then, Theorem~\ref{Corollary:equals} guarantees, under certain assumptions of regularity on the support of $P$, that the solution of \eqref{dual} is unique up to an additive constant.

Note that, in practice, a probability $P\in \mathcal{P}(\mathbb{T}^d)$ defines a periodic measure $\mu_P\in \mathcal{M}(\mathbb{R}^d)$ w.r.t. any $\p\in \Z^d$.  In other words,  $T_{\p} \# \mu_P=\mu_P$,  for all $\p\in \Z^d$, where $T_{\p}:\R^d\rightarrow \R^d$ is the shift operator $\x\mapsto \x+\p.$ A  measure $\mu_P$ is \emph{periodic} if it is the natural extension of some probability measure $P\in \mathcal{P}(\mathbb{T}^d)$. As anticipated, the goal of this section is to show the existence of $\| \cdot \|^p$-cyclically monotone mappings between two periodic measures $\mu_P,\mu_Q\in \mathcal{M}(\mathbb{R}^d)$ absolutely continuous w.r.t. the Lebesgue measure on $\R^d$, denoted as $\mu_P, \mu_Q\ll \ell_d$. 
As commented before, \cite{CORDEROERAUSQUIN1999199} established the existence of a  $\| \cdot \|^2$-cyclically monotone map (which a.s. is the gradient of a convex function $\varphi$) such that $\nabla \varphi \# \mu_P=\mu_Q$.  Theorem 1.25 in \cite{Santambrogio} entails that there is a unique solution of the Monge problem in the torus, described by the relation $T=\x-\nabla f(\x)$, where the sum is to be intended modulo $\Z^d$ and $f$ is an optimal transport potential for the quadratic cost.  Note that this is a quite similar relation (between potentials and transport) to the one in the quadratic transport problem in $\R^d$.  

The proof of Theorem~\ref{TheoremExistenceOfArrangements} starts by realizing that since $\mathbb{T}^d$ is a Polish space, then Theorem 4.1 in  \cite{Villani2008} implies that there exists a solution $\pi^*$ of \eqref{kant}. Furthermore, Theorem 5.10 in \cite{Villani2008} establishes that  $\text{supp}(\pi^*)$ is $d^p$-cyclically monotone, which implies that the set 
\begin{equation}\label{Gamma}
\Gamma=\{ (\x+p,\y+p):\  ( \bx, \by)\in\text{supp}(\pi^*),\ \x\in[0,1]^d, \ d(\bx, \by)=\| \x-\y\|\ \text{and}\ p\in\Z^d\}
\end{equation}
is cyclically monotone. Corollary 3.5 in ~\cite{GangboMccann} implies that this cyclically monotone set is contained in the graph of a $\|\cdot\|^p$-differential
$$\partial^{\|\cdot\|^p}\varphi_p(\x)=\{\by:\ \varphi_p(\z)\leq \varphi_p(\x)+ \|\z-\y\|^p-\|\x-\y\|^p, \text{ for all $\z\in\R^d$}\}$$
of a $\|\cdot\|^p$-concave function  $\varphi_p$ (defined as in \eqref{dconcave} but replacing $d^p$ with $\|\cdot\|^p$). In conclusion, the a.s. uniqueness of this $\|\cdot\|^p$-differential ends the proof.

\begin{Theorem}\label{TheoremExistenceOfArrangements}
Let $P,Q\in \mathcal{P}(\mathbb{T}^d)$ be probability measures such that $\mu_P\ll \ell_d$. Then, there exists a unique solution $\mathbf{T}_p$ of \eqref{monge}. Moreover, there exists a $\mu_P$-a.e. defined $\| \cdot \|^p$-cyclically monotone map $\mathbf{S}_p$ such that
\begin{itemize}
\item the relation $\mathbf{T}_p\circ\tau=\tau\circ(\mathbf{S}_p) $ holds $\mu_P$-almost surely,
\item and $\mathbf{S}_p\# \mu_P=\mu_Q$.
\end{itemize}
\end{Theorem}
The following result gives the uniqueness, up to additive constants, of the optimal transport potential, where the assumptions are given with respect to its associated periodic measures. In particular, we need to have  \emph{negligible boundary} of $\mu_P$ which means that  the boundary of its support has Lebesgue measure $0$, $\ell_d(\partial\,\text{supp}(\mu_P))=0$. The proof investigates the intrinsic relation between the optimal transport potentials and  the previously described $\mathbf{T}_p$, which allows the use of general results for the uniqueness of $\| \cdot \|^p$-concave functions (see \cite{Alberto}) which have the same gradient a.s. in a connected domain of $\mathbb{R}^d$.

\begin{Theorem}\label{Corollary:equals}
Let $P,Q\in \mathcal{P}(\mathbb{T}^d)$  be probability measures with connected support such that their associated periodic measures satisfy $\mu_P, \mu_Q\ll \ell_d$ with negligible boundary. Then, there exists a unique, up to an additive constant,   $d^p$-concave function $f_p$ solution of  \eqref{dual}.
\end{Theorem}
\added{
The assumption of connected support can be relaxed, via \cite[Theorem 2]{Staudt2022OnTU}, to the setting where both measures have disconnected support. If the supports of $\mu_P$ and $\mu_Q$ decompose into closures of connected open components
\begin{equation}
    \label{eq:uniquenessMunk}
    \text{supp}(\mu_P)= \bigcup_{i\in \mathcal{I}}\mathcal{X}_{i,\mu_P}, \quad \text{supp}(\mu_Q)=\bigcup_{j\in \mathcal{J}}\mathcal{X}_{j,\mu_Q}, 
\end{equation}
where $\mathcal{I}$ is finite index set and $\mathcal{J}$ is a countable index set, then, assuming for all non-empty proper $\mathcal{I}'\subset \mathcal{I}$ and $\mathcal{J}'\subset J$ that 
\begin{equation}
    \label{eq:uniquenessMunk2}
    \sum_{i\in \mathcal{I}'}\mu_P(\mathcal{X}_{i,\mu_P} )\neq \sum_{j\in \mathcal{J}'}\mu_Q(\mathcal{X}_{j,\mu_Q} ),
\end{equation}
it follows by \cite[Lemma 5]{Staudt2022OnTU} that no
degenerate transport plan exists. Hence, invoking Theorem 2 in \cite{Staudt2022OnTU} in conjunction with Theorem~\ref{Corollary:equals}, yields an extension of the uniqueness result to measures with disconnected support.}

\begin{Corollary}\label{Corollary:equals2}
\added{
Let $P,Q\in \mathcal{P}(\mathbb{T}^d)$  be probability measures such that their associated periodic measures satisfy $\mu_P, \mu_Q\ll \ell_d$ with negligible boundary where \eqref{eq:uniquenessMunk} and \eqref{eq:uniquenessMunk2} hold. Then, there exists a unique, up to an additive constant,   $d^p$-concave function $f_p$ solution of  \eqref{dual}.
}\end{Corollary}

The importance of Corollary~\ref{Corollary:equals2} mainly lies in that it enables the study of the asymptotic behavior of the potential, allowing us to apply Arzelá-Ascoli like reasoning, as explained in the following section.

\subsection{Asymptotic behaviour}\label{section_asyimptotics}
This section deals with the asymptotic properties of the transport map and potentials.  We consider two sequences of probability measures $\{ \alpha_n\}_{n\in \N}, \{ \beta_n\}_{n\in \N}\subset \mathcal{P}(\mathbb{T}^d)$ converging weakly to $P$ and $Q$ respectively,
\begin{equation*}
\alpha_n\xrightarrow{w} P \ \ \text{and} \ \ \beta_n\xrightarrow{w} Q.
\end{equation*}
Since $\mathbb{T}^d$  is compact,  here the weak convergence is in the sense that for every continuous function $h\in \mathcal{C}(\mathbb{T}^d)$, $\int h(\bx)d\alpha_n(\bx)\rightarrow \int h(\bx)dP(\bx)$. Once again, thanks to that compactness  the existence of moments of any order is always fulfilled for $P\in \mathcal{P}(\mathbb{T}^d) $. As a consequence, Theorem~7.12~in~\cite{Villani2003} implies that $\alpha_n\xrightarrow{w} P$ if and only if the $p$-\emph{Wasserstein distance} $\mathcal{W}_p(\alpha_n,P):=\left(\mathcal{T}_p(\alpha_n,P)\right)^{\frac{1}{p}}$ tends to $0$. An analogous reasoning implies the convergence $\mathcal{T}_p(\alpha_n,\beta_n)\rightarrow\mathcal{T}_p(P,Q)$ for the two-sample case.

The idea of this section is to take advantage of the fact that any $d^p$-concave function $f$ is continuous whereby it is finite. Moreover, it has bounded continuity modulus, so we can  apply Arzelá-Ascoli's Theorem by fixing the constants.
\begin{Lemma}\label{Lemma:lip}
Every $d^p$-concave function $f$ is Lipschitz (in its definition domain $\operatorname{dom}(f)$) with constant $L= 2\,p\, d^{\frac{p-1}{2}}$, with respect to the metric \eqref{ground_distance}. 
\end{Lemma}

The proof of the next Theorem first proceeds by choosing the sequence $\{a_n\}_{n\in \N}$ to guarantee the uniform boundedness of the sequence  $\{(f_n,g_n)\}_{n\in\N}$ of solutions of \eqref{dual}.  This, together with Lemma~\ref{Lemma:lip} and Arzelá-Ascoli's Theorem,  implies that $\{(f_n,g_n)\}_{n\in\N}$ is relatively compact.  The uniqueness of solutions of \eqref{dual}, described in Theorem~\ref{Corollary:equals}, allows us to conclude.
\begin{Theorem}\label{Theoremeq}
Let $P,Q\in \mathcal{P}(\mathbb{T}^d)$  be probability measures with connected supports whose associated periodic measures satisfy $\mu_P, \mu_Q\ll \ell_d$ with negligible boundary.
Let $\{ \alpha_n\}_{n\in \N}$ and $\{\beta_n\}_{n\in \N}\subset \mathcal{P}(\mathbb{T}^d)$ be two sequences of probability measures converging weakly to $P$ and $Q$ respectively.  Denote by  $(f_n,g_n)$ (resp.  $(f,g)$) the solution of the dual problem between $\alpha_n$ and $\beta_n$ (resp. $P$ and $Q$). Then there exists a sequence of real numbers $\{a_n\}_{n\in \N}$ such that $f_n+a_n\rightarrow f$ uniformly on the compact sets of $\mathcal{X}_P$.
\end{Theorem}

\subsection{Asymptotic normality}\label{section:normality}
This section is devoted a proof of a Central Limit Theorem (CLT) for the fluctuations of the empirical optimal transport cost.  Recall that the previous section proves that, under certain regularity assumptions, there exists a unique optimal transport potential from $P$ to $Q$.  Let $f_p$ be such a potential. We will use Efron-Stein's inequality to derive that
\begin{align*}
    \sqrt{n}\left(\mathcal{T}_p(P_n,Q)- \mathbb{E}\mathcal{T}_p(P_n,Q)\right)\stackrel{w}{\longrightarrow} N(0, \sigma^2_p(P,Q)),
\end{align*}
with
\begin{equation}\label{sigma}
\sigma^2_p(P,Q)=\operatorname{Var}(f_p(X)). 
\end{equation}
Then, we will see that the same holds in the two sample case.
The idea is not new: it has already been used with the same goal in \cite{DelBarrio2019} for the quadratic cost in $\R^d$, and in its extension to general costs in \cite{Alberto}. Moreover,  when using regularized optimal transport,  \cite{Mena2019StatisticalBF} showed that the same technique can be applied. A similar result, but using the idea in \cite{delbarrio_semidiscreto} of differentiating the supremum in the functional sense by applying the general result of \cite{carcamo}, yields also a CLT on the torus for $p\geq 2$, see \cite{Munk_2022}.

\begin{Theorem}\label{Theo:centrallim2}
Let $P,Q\in \mathcal{P}(\mathbb{T}^d)$  be probability measures with connected supports such that their associated periodic measures satisfy $\mu_P, \mu_Q\ll \ell_d$ with negligible boundary. Then, for any $p>1$, we have
\begin{align*}
    \sqrt{n}\left(\mathcal{T}_p(P_n,Q)- \mathbb{E}\mathcal{T}_p(P_n,Q)\right)\stackrel{w}{\longrightarrow} N(0, \sigma_p^2(P,Q)),
\end{align*}
and, if $m=m(n)$ satisfies that $m\longrightarrow+\infty$ and $\frac{n}{n+m} \rightarrow \lambda \in (0,1)$ as $n\rightarrow \infty$,
\begin{equation*}\label{tcl}
\sqrt{\textstyle \frac{nm}{n+m}}\left(\mathcal{T}_p(P_n, Q_m)-\mathbb{E}\mathcal{T}_p(P_n, Q_m)\right)\stackrel{w}{\longrightarrow} N\left(0,(1-\lambda)\sigma^2_p(P,Q)+\lambda\sigma^2_p(Q,P)\right),
\end{equation*}
where $\sigma^2_p(P,Q)$ and  $\sigma^2_p(Q,P)$ are defined in \eqref{sigma} and satisfy
\begin{equation}\label{variance_tcl}
\sqrt{\textstyle \frac{nm}{n+m}}\mathrm{Var}(\mathcal{T}_p(P_n,Q_m))\longrightarrow (1-\lambda)\sigma^2_p(P,Q)+\lambda\sigma^2_p(Q,P).
\end{equation}
\end{Theorem}
It is clear that the limit of Theorem~\ref{Theo:centrallim2} degenerates to $0$ when $P=Q$. Suppose now that $P\neq Q$ are satisfying the assumption of Theorem~\ref{Theo:centrallim2}. The limit, in the one-sample case, is degenerate if and only if $\operatorname{Var}(f_p(X))=0$. Since the optimal transport potentials are unique up to additive constants, see Theorem~\ref{Corollary:equals}, we can suppose that $E(f_p(X))=0$. Thus, the degeneracy is equivalent to $E(f_p(X)^2)=0$, hence $f_p=0$ $P$-a.s. and the same holds for  $f_p^{d^p}$. This implies, in particular, that $\mathcal{W}_p(P,Q)=0$ which occurs only if $P=Q$.

Our initial motivation to prove Theorem \ref{Theo:centrallim2} was to find an asymptotic distribution of $\mathcal{T}_p(P_n,Q_m)$ allowing the definition of a two-sample goodness-of-fit test. Even for measures supported on the real line, the only asymptotic results account for the case $P\neq Q$, providing the asymptotic behaviour of the Wasserstein statistic under the alternative hypothesis. The idea of switching $H_0$ and $H_1$ and testing for similarities has been studied in several previous works, all considering measures supported on $\mathbb{R}$. Gaussian deviations from the true distance $\mathcal{T}_2(P,Q)$ are proved in \cite{Gordaliza2019}, which allows testing of $\mathcal{T}_2(P,Q)\geq \Delta_0$, for a given threshold $\Delta_0$. In the same way, the earlier work \cite{Freitag2007} introduced such an asymptotic test for assessing similarities based on the trimmed Wasserstein distance, allowing sample dependency. 

Unfortunately, the same strategy can not be applied in our case, as the derived CLT for measures supported on $\mathbb{T}^2$ (Theorem~\ref{Theo:centrallim2}) only states Gaussian deviations from the mean. Indeed, if we use (\ref{variance_tcl}), we could consider the statistic
\begin{equation}\label{asymptotic_statistic}
\frac{\mathcal{T}_p(P_n,Q_m)-\mathbb{E}\mathcal{T}_p(P_n,Q_m)}{\sqrt{\mathrm{Var}(\mathcal{T}_p(P_n,Q_m))}}\overset{w}{\underset{P\neq Q}{\longrightarrow}} N(0,1),
\end{equation}
where, in practice, the variance and expectation could be estimated by bootstrapping the given samples (as long as bootstrap consistency is ensured). \added{The recent works of \cite{Hundrieser2022AUA} and \cite{manole2021plugin} show that, in small dimension --$d=2,3$ and at most $4$--, the value $\mathbb{E}\mathcal{T}_p(P_n,Q_m)$ can be substituted by the population $\mathcal{T}_p(P,Q)$. That gives rise to 
\begin{equation}\label{asymptotic_statistic_correctCent}
\frac{\mathcal{T}_p(P_n,Q_m)-\mathcal{T}_p(P,Q)}{\sqrt{\mathrm{Var}(\mathcal{T}_p(P_n,Q_m))}}\overset{w}{\underset{P\neq Q}{\longrightarrow}} N(0,1),
\end{equation}
see \cite[Example 5.7]{Hundrieser2022AUA} for general $p$ or \cite[Corollary 8]{manole2021plugin} for $p=2$.
However,  for dimension $d>4$ and $p=2$, this substitution is no longer valid \cite[Proof of Proposition 21]{Manole2021SharpCR}. }

When $P\neq Q$, the statistics in \eqref{asymptotic_statistic}, \eqref{asymptotic_statistic_correctCent} converge in law to a standard Gaussian distribution. This is illustrated in Figure \ref{fig:tcl}. However, one would expect the statistic to be stochastically larger under $P\neq Q$ than under $P=Q$, allowing the distinction of the null and the alternative hypotheses. Nevertheless, due to the aforementioned degeneracy of Theorem \ref{Theo:centrallim2} when $P=Q$, this condition fails to be satisfied and no asymptotic test can be implemented from this result. Further discussion about this issue can be found in Section \ref{sec:discussion}. 
Therefore, the rest of this paper is devoted to alternative approaches to define suitable two-sample goodness-of-fit tests for measures supported on $\mathbb{T}^2$.

\section{Two-sample goodness-of-fit tests}\label{section:test}

Let us first formulate the problem. Denote by $(X_1,\ldots,X_n)$ and $(Y_1,\ldots,Y_m)$ two independent and identically distributed random samples of laws $P,Q\in\mathcal{P}(\mathbb{T}^2)$ respectively, and by $P_n$, $Q_m$ their corresponding empirical probability measures. We aim to test
\begin{equation}\label{null_torus}
H_0: P=Q\qquad\textrm{against}\qquad H_1: P\neq Q
\end{equation}
via the the definition of a statistic $T_{nm}=T(P_n,Q_m)$, representing an estimate of discrepancy between $P_n$ and $Q_m$, together with the critical region
\begin{equation}
\label{eq:crit-region}
R=\lbrace (x_1,\ldots,x_n;y_1,\ldots,y_m)\,:\, T_{nm}\geq c_{nm}(\alpha)\rbrace,
\end{equation}
where $x_i$ (resp. $y_j$) denotes a realization of $X_i$ (resp. $Y_j$) for $i=1,\ldots,n$ (resp. $j=1,\ldots,m$). The critical value $c_{nm}(\alpha)$ in \eqref{eq:crit-region} is given for a fixed significance level $\alpha$ by
\begin{equation}
c_{nm}(\alpha)=\inf\lbrace t>0\,:\, F_{nm}(t)\geq 1-\alpha\rbrace,
\end{equation}
where $F_{nm}$ is the distribution function of the statistic $T_{nm}$ under $H_0$. We are therefore considering the test
\begin{equation}\label{ideal_test}
\pi_{nm}=\left\{ \begin{array}{lcc}
             1 &  \textrm{if} & \quad T_{nm}\geq c_{nm}(\alpha)   \\
             0 &   \textrm{otherwise}
             \end{array}
   \right.
\end{equation}
Equivalently, a $p$-value for this test is $p_{nm} = 1- F_{nm}(T_{nm})$.
Ideally, we would like $T_{nm}$ to be $\mathcal{T}_p(P_n,Q_m)$. However, knowing the distribution of the latter statistic under $H_0$ remains an open problem. The one-sample case in $\mathbb{R}^d$ has recently been addressed in \cite{Hallin2021}, but approaches for two-sample testing in arbitrary dimension, and for measures on more general spaces, have not already been proposed to the best of our knowledge. The lack of solutions may be explained by the intrinsic difficulty of characterizing the distribution of $\mathcal{T}_p(P_n,Q_m)$ when $P=Q$ especially when the dimension is larger than one. In the next subsections, we propose two alternative approaches to define \eqref{ideal_test}, both based on the $2$-Wasserstein distance, that allow two-sample goodness-of-fit testing for measures on $\mathbb{T}^2$.

\subsection{Geodesic projections into \texorpdfstring{$\mathbb{R}/\mathbb{Z}$}{the circle}} \label{sec:test_geodesics}

Our first approach for testing the equality of two measures $P$, $Q$ on $\mathbb{R}^2/\mathbb{Z}^2$ is to test the equality of their geodesic projections. This bypasses the dimension problem and allows the implementation of testing techniques based on Wasserstein distance for one-dimensional spaces. Geodesics on $\mathbb{T}^2$ are the images by the canonical projection $\tau$ of straight lines on $\mathbb{R}^2$ \cite{boothby}. Lines with irrational slope map to geodesics which are dense on $\mathbb{T}^2$, and only lines with rational slope map to \textit{closed} geodesics on the torus, which are closed spirals isomorphic to $\mathbb{R}/\mathbb{Z}$ (see \cite[Figure VII.10]{boothby} for an illustration).

The strategy is to project $P_n$ and $Q_m$ into $N_g$ closed geodesics, and to test the equality of each pair of projected measures, which will be supported on $\mathbb{R}/\mathbb{Z}$. These geodesics can be chosen a priori by the practitioner, or sampled from the set of all closed geodesics on $\mathbb{T}^2$. We propose a sampling method in Appendix \ref{appendix:geodesic_sampling}. This method prioritizes simpler geodesics (that is, with a smaller number of revolutions over the torus) in order to ease computational implementation. The algorithm we used to project samples on $\mathbb{T}^2$ to a given geodesic is described in Appendix \ref{appendix:geodesic_projection}. To avoid repetition of the same test, and to ensure independence between the computed $p$-values, we require all the $N_g$ geodesics to be different.

In this section, we propose a two-sample Wasserstein test to assess the equality of two measures supported on the circle, and state how to combine the resulting $N_g$-tuple of $p$-values into a global $p$-value for the bi-dimensional problem. From now on, to simplify notation, we will denote by $\mathcal{T}_2$ any squared Wasserstein distance, the ground space being inferred from the corresponding measures.

\subsubsection{Two-sample goodness-of-fit test on \texorpdfstring{$\mathbb{R}/\mathbb{Z}$}{the circle}}\label{sec:circle_test}

Optimal Transport on the circle has been recently studied in detail in \cite{hundrieser2021}, where the limit laws of the one and two-sample empirical Wasserstein distance for measures on $\mathbb{R}/\mathbb{Z}$ are derived. However, the considered statistics are not distribution-free, so that only one-sample goodness-of-fit tests can be derived from these results. Still, the authors of \cite{hundrieser2021} also propose a $b$-out-of-$n$ bootstrap approach, for $b=o(n)$, to define a two-sample goodness-of-fit test. Unfortunately, type I error fails to be controlled since the bootstrapped $p$-value under the null hypothesis is (substantially) stochastically smaller than a uniform random variable. This can be observed by simple numerical experiments based on the implementation proposed by \cite{hundrieser2021}, for example by comparing two equally-sized samples from a Uniform distribution. We believe that this is due to a lack of consistency of the two-sample bootstrap for the proposed statistic. In order to bypass this issue, we now propose a convenient alternative approach based on a distribution-free two-sample statistic.

Let $P^c,Q^c\in\mathcal{P}(\mathbb{R}/\mathbb{Z})$ and $P^c_n,Q^c_m$ be their corresponding empirical probability measures. We aim to test
\begin{equation*}
    H_0:P^c=Q^c\qquad\textrm{against}\qquad H_1:P^c\neq Q^c.
\end{equation*}
\added{If $\mathbb{R}/\mathbb{Z}$ is parameterized by the set $[0,1)$ with the geodesic distance 
\begin{equation*}
    d_{\mathbb{R}/\mathbb{Z}}(x,y) = \min\lbrace |x-y|, 1-|x-y|\rbrace,
\end{equation*}
the cumulative distribution functions of $P^c, Q^c$, denoted as $F$, $G$ respectively, can be defined as in \cite{hundrieser2021} as
\begin{equation}
    F(t)=P^c([0,t]),\quad G(t)=Q^c([0,t])\qquad\forall\,t\in[0,1).
\end{equation}
}Then, we can write
\begin{equation}\label{w_circle}
    \mathcal{T}_2(P^c,Q^c)=\underset{\alpha\in\mathbb{R}}{\inf}\int_0^1\left(F^{-1}(t)-(G-\alpha)^{-1}(t)\right)^2\,dt,
\end{equation}
where the pseudo-inverse is defined as $H^{-1}(s)=\inf\lbrace t\,:\,H(t)>s\rbrace,$ for any distribution function $H$. The formulation \eqref{w_circle} was first proved in \cite{Rabin2009} for discrete measures, and extended to arbitrary measures in \cite{delon_circle}. It shows how the Optimal Transport problem on the circle reduces to the same problem on $[0,1)\subset\mathbb{R}$ if both measures are relocated on the real line choosing as origin the minimizing element $\alpha$. This is well illustrated in \cite{hundrieser2021}. We first remark that if one of the two measures is the uniform law on $\mathbb{R}/\mathbb{Z}$, the infimum on \eqref{w_circle} has an explicit formulation.

\begin{Lemma}\label{lemma_uniform} Let $P^c\in\mathcal{P}(\mathbb{R}/\mathbb{Z})$, and $F$ be its cumulative distribution function. Let $U$ be the uniform distribution on $\mathbb{R}/\mathbb{Z}$. Then,
\begin{equation*}
    \mathcal{T}_2(P^c,U)=\int_0^1\left(F^{-1}(t)-t-\alpha_0(F)\right)^2\,dt,
\end{equation*}
where the optimal origin is given by
\begin{equation*}
    \alpha_0(F)=\int_0^1(F^{-1}(t)-t)\,dt.
\end{equation*}
\end{Lemma}
If we replace $P^c$ and $F$ by their empirical counterparts, $P_n^c$ and $F_n$, Lemma \ref{lemma_uniform} allows the definition of the statistic $\mathcal{T}_2(P^c_n,U)$, which is distribution-free when $P^c=U$. 

\begin{Lemma}\label{lemma:free_null_distribution} Let $P^c\in\mathcal{P}(\mathbb{R}/\mathbb{Z})$, $P^c_n$ be its empirical probability measure, and $U$ be the uniform distribution on $\mathbb{R}/\mathbb{Z}$. Then, if $P^c=U$,

\begin{equation*}
    n\,\mathcal{T}_2(P^c_n,U)\overset{w}{\underset{n}{\longrightarrow}}\int_0^1\mathbb{B}(t)^2\,dt-\left(\int_0^1 \mathbb{B}(t)\,dt \right)^2,
\end{equation*}
where $\mathbb{B}$ is a standard Brownian bridge, and the weak convergence is understood as convergence of probability measures on the space of right-continuous functions with left limits.
\end{Lemma}

Lemma \ref{lemma:free_null_distribution} can be used to define of a one-sample goodness-of-fit uniformity test, based on the squared Wasserstein distance on the circle. This would complement the work in \cite{hundrieser2021}, where such a test was introduced for the $1$-Wasserstein distance. As our aim here is to define a two-sample test, we adapt the idea of \cite{Ramdas2015} to compare two measures on the circle, by considering the $2$-Wasserstein distance between $G_m^{-1}(F_n)$ and the uniform distribution. We can therefore consider the statistic
\begin{gather}
    T^c_{nm}=\frac{nm}{n+m}\mathcal{T}_2(G_m\#P^c_n,U)=\label{circle_statistic}\\\frac{nm}{n+m}\int_0^1 \left(G_m(F_n^{-1}(t))-t-\alpha_0(F_n^{-1}( G_m))\right)^2dt,\nonumber
\end{gather}
which is also distribution-free when $P^c=Q^c$. The following result is the counterpart of Lemma \ref{lemma:free_null_distribution}.

\begin{Proposition}\label{prop:free_circle_statistic} Let $P^c,Q^c\in\mathcal{P}(\mathbb{R}/\mathbb{Z})$, \added{having continuous and strictly increasing cumulative distribution functions}. Let $P^c_n,Q^c_m$ be their corresponding empirical probability measures, and $F_n, G_m$ be their empirical cumulative distribution functions. If $\frac{n}{m}\rightarrow\lambda$ when $n,m\rightarrow\infty$ for some $\lambda\in[0,\infty)$ then, under $P^c=Q^c$, it holds that
\begin{equation*}
    T_{nm}^c=\frac{nm}{n+m}\mathcal{T}_2(G_m\#P^c_n,U)\overset{w}{\underset{n,m}{\longrightarrow}}\int_0^1\mathbb{B}(t)^2\,dt-\left(\int_0^1 \mathbb{B}(t)\,dt \right)^2.
\end{equation*}
\end{Proposition}
Consequently, with the notation of the beginning of Section \ref{section:test}, we propose the test
\begin{equation}\label{marginal_test}
\pi_{nm}^c=\left\{ \begin{array}{lcc}
             1 &  \textrm{if} & \quad T^c_{nm}\geq c^c_{nm}(\alpha) \\
             0 &   \textrm{otherwise}
             \end{array}
   \right.
\end{equation}
where the critical value $c_{nm}^c(\alpha)$ is  given by
\begin{equation*}
    c_{nm}^c(\alpha)=\inf\left\lbrace t>0\,:\, F_{nm}^c(t)\geq 1-\alpha\right\rbrace\,,
\end{equation*}
with $F_{nm}^c$ denoting the distribution function of $T_{nm}^c$ under $H_0$. 
Equivalently, a $p$-value for this test is $p^c_{nm} = 1- F^c_{nm}(T^c_{nm})$.
Following Proposition \ref{prop:free_circle_statistic}, the critical value or, equivalently, the $p$-value for a given sample, can be approximated with arbitrary precision using a Monte Carlo algorithm. The following result guarantees the consistency of \eqref{marginal_test}.

\begin{Proposition}[Consistency]\label{prop:consistency_marg} Let $P^c,Q^c\in\mathcal{P}(\mathbb{R}/\mathbb{Z})$ \added{having continuous and strictly increasing cumulative distribution functions}. If $P^c\neq Q^c$, it holds
\begin{equation*}
    \underset{n,m\rightarrow\infty}{\lim}\mathbb{P}\left(\pi_{nm}^c=1\right)=1\quad\textrm{for any }\alpha>0.
\end{equation*}
\end{Proposition}


\subsubsection{Combining a $N_g$-tuple of tests on \texorpdfstring{$\mathbb{R}/\mathbb{Z}$}{the circle}}\label{sec:ng_geodesics}

Consider the problem of testing the equality of $N_g$ pairs of projections of $P_n$ and $Q_m$ into $N_g$ different closed geodesics. Instead of a single statistic, we now have a sample $(T_{nm,1}^c,\ldots,T_{nm,N_g}^c)$ of statistics which, under the null hypothesis, are identically distributed as $T_{nm}^c$ (by Proposition \ref{prop:free_circle_statistic}). Equivalently, one can think of a sample of $p$-values
$(p_1,\ldots,p_{N_g})$ which, following \eqref{marginal_test}, are given by 
\begin{equation}
  p_i=1-F_{nm}^c(T_{nm,i}^c)\quad i=1,\ldots,N_g.
\end{equation}
\added{These individual $p$-values can be aggregated as follows:
\begin{equation}\label{pvalue_twomarg}
    p^{N_g} = N_g\,\min_{i=1}^{N_g} p_i.
\end{equation}
This aggregation is akin to the Bonferroni correction for Family Wise Error Rate (FWER) control in multiple testing \cite{Bonferroni}. As such, $p^{N_g}$ defined in \eqref{pvalue_twomarg} is a valid $p$-value for the two-dimensional test, regardless of the possible dependencies between the $N_g$ individual $p$-values. This implies that the two-dimensional test}
\begin{equation}\label{torus_marg_test}\tag{$N_g$-geod}
\pi_{nm,N_g}^{g}=\left\{ \begin{array}{lcc}
             1 &  \textrm{if} & p^{N_g}\leq \alpha\\
             0 &   \textrm{otherwise}
             \end{array}\right.
\end{equation}
\added{controls the type I error for any $\alpha>0$ (see Appendix \ref{proofs_3} for a proof)}. Regarding consistency under fixed alternatives, 
by construction, \eqref{torus_marg_test} will fail to detect differences between two measures on $\mathbb{T}^2$ whose projected distributions are identical for all the $N_g$ geodesics considered. Therefore, $\pi_{nm,N_g}^{g}$ will not be consistent under such alternatives, which, are arguably very unlikely in practice if $N_g$ is large enough. Otherwise, consistency is guaranteed. 

\begin{Proposition}[Consistency]\label{prop:consistency_twomarg}Let $P,Q\in\mathcal{P}(\mathbb{T}^2)$ \added{such that $\mu_P,\mu_Q\ll\ell_2$} and $P^c_i$ (resp. $Q^c_i$), $i=1,\ldots,N_g$, be the circular projected distributions of $P$ (resp. $Q)$ to $N_g$ closed geodesics of $\mathbb{T}^2$. If $P^c_i\neq Q^c_i$ for at least one $i\in\lbrace 1,\ldots,N_g\rbrace$, it holds
\begin{equation*}
    \underset{n,m\rightarrow\infty}{\lim}\mathbb{P}\left( \pi_{nm,N_g}^{g}=1\right)\quad\textrm{for any }\alpha>0.
\end{equation*}
\end{Proposition}

\added{
\begin{Remark}\label{remark_abscont} The assumption in Proposition~\ref{prop:free_circle_statistic} that the projected measure $P^c\in\mathcal{P}(\mathbb{R}/\mathbb{Z})$ has continuous and strictly increasing cumulative distribution function is satisfied if the underlying measure $P\in\mathcal{P}(\mathbb{T}^2)$ satisfies $\mu_P\ll\ell_2$. See Appendix \ref{proofs_3} for a proof.
\end{Remark}}

\added{The time complexity of \eqref{torus_marg_test} is $\mathcal{O}(n+m)$. Indeed,  $n+m$ operations are needed to compute $G_m(F_n^{-1}(t))$ and $F_n^{-1}(G_m(t))$ for a given $t$. Therefore, computing the test statistic \eqref{circle_statistic} can be done in $\mathcal{O}(n+m)$ operations, where the complexity constant depends on the number of subdivisions of $[0,1]$ set by the numerical integration method chosen to compute \eqref{circle_statistic}. Moreover, the time complexity of the algorithm described in Appendix \ref{appendix:geodesic_sampling} to sample closed geodesics is also $\mathcal{O}(n+m)$ in practice, as a consequence of the distribution from which the geodesics are drawn. This is empirically illustrated in Figure \ref{fig:time}.
}

\subsection{$p$-value upper bounding}\label{sec:upper_bound}

If we set $\mathcal{T}_2(P_n,Q_m)$ as the statistic $T_{nm}$ for the test \eqref{ideal_test}, the p-value for a given sample would be given by
\begin{equation}\label{ideal_pvalue}
\mathbb{P}_{H_0}(\mathcal{T}_2(P_n,Q_m)\geq t_{nm}),
\end{equation}
where $t_{nm}$ denotes the statistic realization. The goal of this section is to find an upper bound for \eqref{ideal_pvalue}, which will itself be a valid $p$-value for \eqref{ideal_test} if it controls type I error (that is, if it remains with probability $1-\alpha$ over a fixed significance level $\alpha$ under $H_0$). We will also require the power of the corresponding test to tend to 1 under fixed alternatives. We start by upper bounding the deviations of the statistic from the mean. Using McDiarmid's inequality \cite{mcdiarmid}, we obtain the following result, which extends to the two-sample case the inequality in \cite[Proposition 20]{weed_bach}, for the quadratic cost.



\begin{Theorem}\label{Theorem::bound} Let $P,Q\in\mathcal{P}(\mathbb{T}^2)$ and $P_n,Q_m$ be two empirical probability measures of laws $P$, $Q$ respectively. Then, for all $t\in\mathbb{R}$, we have
\begin{equation}\label{bound_mean_dev}
\mathbb{P}\left(\mathcal{T}_2(P_n,Q_m)-\mathbb{E}\mathcal{T}_2(P_n,Q_m)>t\right)\leq \exp\left(-\frac{nm}{n+m} 8t^2 \right).
\end{equation}
\end{Theorem}

After that, we study the convergence speed of the expectation under the null hypothesis. Using directly the results exposed in \cite{Fournier2016}, only bounds of order
\begin{equation}\label{bound_fournier}
\mathbb{E}\mathcal{T}_2(P_n,Q_m)=O\left(n^{-\frac{1}{2}}+m^{-\frac{1}{2}}\right)
\end{equation}
can be expected. However, the recent work in \cite{ambrosio2021quadratic} shows that the convergence of the mean \eqref{bound_fournier} becomes faster under some regularity assumptions. On the one hand, we require the density of the induced periodic measure $\mu_P$ to be Hölder continuous\footnote{A function $f:\mathbb{R}^n\rightarrow\mathbb{R}$ is said to be locally Hölder continuous in a compact set $X$ for some $\alpha>0$ if, for every $x\in X$, there exists some $\epsilon>0$ such that  $|f(x)-f(y)|\leq C \|x-y\|^{\alpha}$ if $y\in X$ and $\|y-x\|<\epsilon$.} and absolutely continuous w.r.t. the Lebesgue measure $\ell_2$ in $\R^2$. On the other hand, we require the set $\operatorname{supp}(P)$ to be connected and to have $\mathcal{C}^1$ boundary, in the sense that it can be locally parameterized by a $\mathcal{C}^1$ curve.
\begin{assumption}\label{Lip} (1) $P\in \mathcal{P}(\mathbb{T}^2)$ is supported in a connected set with $\mathcal{C}^1$  boundary,  with $\mu_P\ll \ell_2$. (2) Its probability density $p$ is Hölder continuous and bounded from below in its support ($p(x)\geq \lambda>0$ for all $x\in \operatorname{supp}(\mu_P)$).
\end{assumption}

If Assumption \ref{Lip} is satisfied, then from Lemma B.1 and Theorem 6.3. in \cite{ambrosio2021quadratic} we can derive the following asymptotic bound for the two-sample null expectation.

\begin{Lemma}\label{Lemma:twosample_exp} Let $P=Q\in\mathcal{P}(\mathbb{T}^2)$ satisfy Assumption \ref{Lip} and $m=m(n)$ be a sequence such that $m\xrightarrow[n\rightarrow\infty]{}\infty$ and $\frac{n}{m}\rightarrow \lambda\in (0,1)$. Then, we have 
\begin{equation}\label{two_sample_bound}
\underset{n\rightarrow\infty}{\lim\sup}\,\frac{n}{\log(n)}\,\mathbb{E}\mathcal{T}_2(P_n,Q_m)\leq \frac{1}{4\pi}\left(1+\frac{1}{\lambda} \right).
\end{equation}
\end{Lemma}

Note that Assumption \ref{Lip} is not especially restrictive. It is satisfied by any continuously differentiable density, whose connected support can be locally given by the graph of a continuously differentiable function. Examples include bivariate von Mises distributions or uniform distributions in connected smooth sets.

The idea to define the test is to combine Theorem \ref{Theorem::bound} with Lemma \ref{Lemma:twosample_exp} and upper bound \eqref{ideal_pvalue} for sufficiently large sample sizes. If we take the limit for the expectation in \eqref{bound_mean_dev} under the null, we have the following result.

\begin{Proposition}\label{prop:epsilon} Let $P,Q\in\mathcal{P}(\mathbb{T}^2)$ and $P_n,Q_m$ be two empirical probability measures of laws $P$, $Q$ respectively. For all $\varepsilon>0$, there exists $N_\varepsilon\in\mathbb{N}$ such that for all $n,m\geq N_\varepsilon$, we have
\begin{equation}\label{bound_mean_dev_epsilon}
\mathbb{P}_{H_0}\left(\mathcal{T}_2(P_n,Q_m)>t\right)\leq \exp\left(-\frac{nm}{n+m} 8(t-\varepsilon)^2 \right)\,=:\,\xi_{nm,\varepsilon}(t)\qquad\forall\,t>0.
\end{equation}
\end{Proposition}
For a fixed $\varepsilon>0$, the bound \eqref{bound_mean_dev_epsilon} can be used to define a test \eqref{ideal_test} for any $\alpha>0$ as follows:
\begin{equation}\label{test_epsilon}\tag{UB}
    \pi_{nm,\varepsilon}^{ub}=\left\{ \begin{array}{lcc}
             1 &  \textrm{if} & \xi_{nm,\varepsilon}\left(\mathcal{T}_2(P_n,Q_m)\right)\leq \alpha\\
             0 &   \textrm{otherwise}
             \end{array}\right.
\end{equation}
By Proposition \ref{prop:epsilon}, the test \eqref{test_epsilon} will control type I error for all $n,m\geq N_\epsilon$. In practice, the threshold $N_\varepsilon$ depends on the unspecified constant hidden in \eqref{two_sample_bound}, which is dragged from the results in \cite{ambrosio2021quadratic}. The following result shows that, nevertheless, asymptotic consistency at level $\alpha$ of \eqref{test_epsilon} is guaranteed.
\begin{Proposition}[Asymptotic consistency at level $\alpha$]\label{prop:asymptotic_consistency} Let $P,Q\in\mathcal{P}(\mathbb{Z}^2)$. The test \eqref{test_epsilon} is asymptotically of level $\alpha$. If $P=Q$, we have, for any $\varepsilon>0$,
\begin{equation}
   \underset{n,m\rightarrow\infty}{\lim} \mathbb{P}\left(\pi_{nm,\varepsilon}^{ub}=1\right)\leq \alpha\quad\textrm{for any }\alpha>0.
\end{equation}
Under fixed alternatives, \added{the test is consistent if $\mathcal{T}_2(P,Q)>\varepsilon$}:
\begin{equation*}
    \underset{\begin{smallmatrix} n,m \rightarrow \infty \end{smallmatrix}}{\lim}\mathbb{P}\left(\pi_{nm,\varepsilon}^{ub}=1\right)=1\quad\textrm{for any }\alpha>0.
\end{equation*}
\end{Proposition}

\added{The last result ensures asymptotic consistency at level $\alpha$ if the two compared measures are further than $\varepsilon$ in the squared $2$-Wasserstein distance. This can be used to calibrate the sensibility of \eqref{test_epsilon} if the practitioner possesses some prior information about the differences that the test should accept. This would ensure smaller $N_\varepsilon$ without implying a power decrease. For the simulation and case studies presented here, we will set $\varepsilon$ to the machine precision  $\varepsilon_m=2.2\cdot 10^{-16}$ (for a standard double-precision floating-point format). The corresponding $N_\varepsilon$ should be affordable thanks to Lemma \ref{Lemma:twosample_exp}, responsible of the satisfactory power of \eqref{test_epsilon}.} Due to the improved convergence speed of the expectation, we will have sharp bounds \eqref{bound_mean_dev_epsilon} for reasonable sample sizes, allowing the detection of differences for our practical purposes. This is illustrated in Section \ref{section::asymptotic_performance}.

\added{The computational complexity of \eqref{test_epsilon} is given by the numerical algorithm solving the Optimal Transport problem. Here, we used the Fast Network Simplex for Optimal Transport \cite{lemon}, which has $\mathcal{O}((n+m)^2)$ time complexity and $\mathcal{O}((n+m)^2)$ memory cost, due to the cost matrix computation.}

\section{Numerical experiments}\label{simulations}

This section is devoted to assess the performance of the two-sample goodness-of-fit tests \eqref{torus_marg_test} and \eqref{test_epsilon}, and to show how they can be implemented to evaluate differences on protein structure data. In Section \ref{section::small_sample_performance}
and \ref{section::asymptotic_performance}, we evaluate the relative efficiency of both tests, comparing their performance with other methods not based on Optimal Transport. Section \ref{section::proteins} illustrates one possible application to protein structure investigations, by stating statistical evidence of nearest neighbors effects on local protein conformations.

\subsection{Small-sample performance}\label{section::small_sample_performance}

To make an informative analysis of the performance of tests \eqref{torus_marg_test} and \eqref{test_epsilon}, we studied how their power function behaves for alternatives converging to the null hypothesis. We also assessed whether the proposed approach to define a Wasserstein test on the circle contributes to a better power. In particular, we compared the power function of \eqref{torus_marg_test} with variations of the same test. On the one hand, to evaluate whether the choice of an optimal origin to relocate the measures on $[0,1)$ is advantageous, we considered the same statistic \eqref{circle_statistic} but with $\alpha_0$ being random and uniformly chosen in $[0,1]$. It is easy to check that the modified statistic is distribution-free under the null, by proceeding analogously to Proposition \ref{prop:free_circle_statistic}. On the other hand, to study whether the use of Wasserstein distance for the one-dimensional statistic contributes to a better power, we relocated the measures in $[0,1)$ (again after choosing a random origin on $\mathbb{R}/\mathbb{Z})$ and compared them with the well-known Anderson-Darling two-sample statistic. To study the effect of the number $N_g$ of geodesics, we performed the test \eqref{torus_marg_test} for $N_g\in\lbrace 2,3,4,5\rbrace$. We also compared all the previous approaches with the two-dimensional extension of the Kolmogorov-Smirnov two-sample test proposed by Fasano and Franceschini \cite{fasano_franceschini}, defined for measures supported on $\mathbb{R}^2$. This allows the assessment of whether taking into account the geometry of the underlying space contributes to a better performance.

For the small-sample case, we compared samples of size $n=m=50$ drawn from a bivariate von Mises (bvM) distribution \cite{prot_vonmises} of means $\mu=\nu=0.5$, and concentration parameters $\kappa_1,\kappa_2,\kappa_3$ with equally-sized samples drawn from a uniform distribution on $\mathbb{T}^2$. \added{The density of the bvM cosine model is given by
\begin{align*}
    f(\varphi,\psi)=c(\kappa_1,\kappa_2,\kappa_3)(\exp(\kappa_1\cos(\varphi-\mu)+\kappa_2\cos(\phi-\nu)-\\ 
    \kappa_3\cos(\varphi-\mu-\psi+\nu)),
\end{align*}
where the explicit form of the normalization constant $c(\kappa_1,\kappa_2,\kappa_3)$ is stated in \cite{prot_vonmises}.} The null hypothesis corresponds to the case $\kappa_1=\kappa_2=\kappa_3=0$.
For the converging alternatives, we distinguished two scenarios: 
\begin{itemize}
    \item[(a)] No dependence structure: $\kappa_3=0$ and $\kappa_1=\kappa_2\in[0,3]$ as varying parameter.
    \item[(b)] Only dependence structure: $\kappa_1=\kappa_2=0$ and $\kappa_3\in[0,3]$ as varying parameter. Here, the marginal laws are uniform distributions on $[0,1]$ \cite{prot_vonmises}.
\end{itemize}
The rejection probability was estimated as the proportion of rejections at level $\alpha=0.05$ among $5000$ repetitions of each test for a fixed value of the corresponding varying parameter. Results for both scenarios are shown in Figure \ref{fig:small_sample}, where `W-geodesic' stands for the test \eqref{torus_marg_test}, `Naive W-geodesic' for its random origin variation, `AD-geodesic' for the comparison with the Anderson-Darling two-sample statistic, and `Upper bound' for the test \eqref{test_epsilon}.

\begin{figure}[htp]
    \centering
    \includegraphics[scale = 0.24]{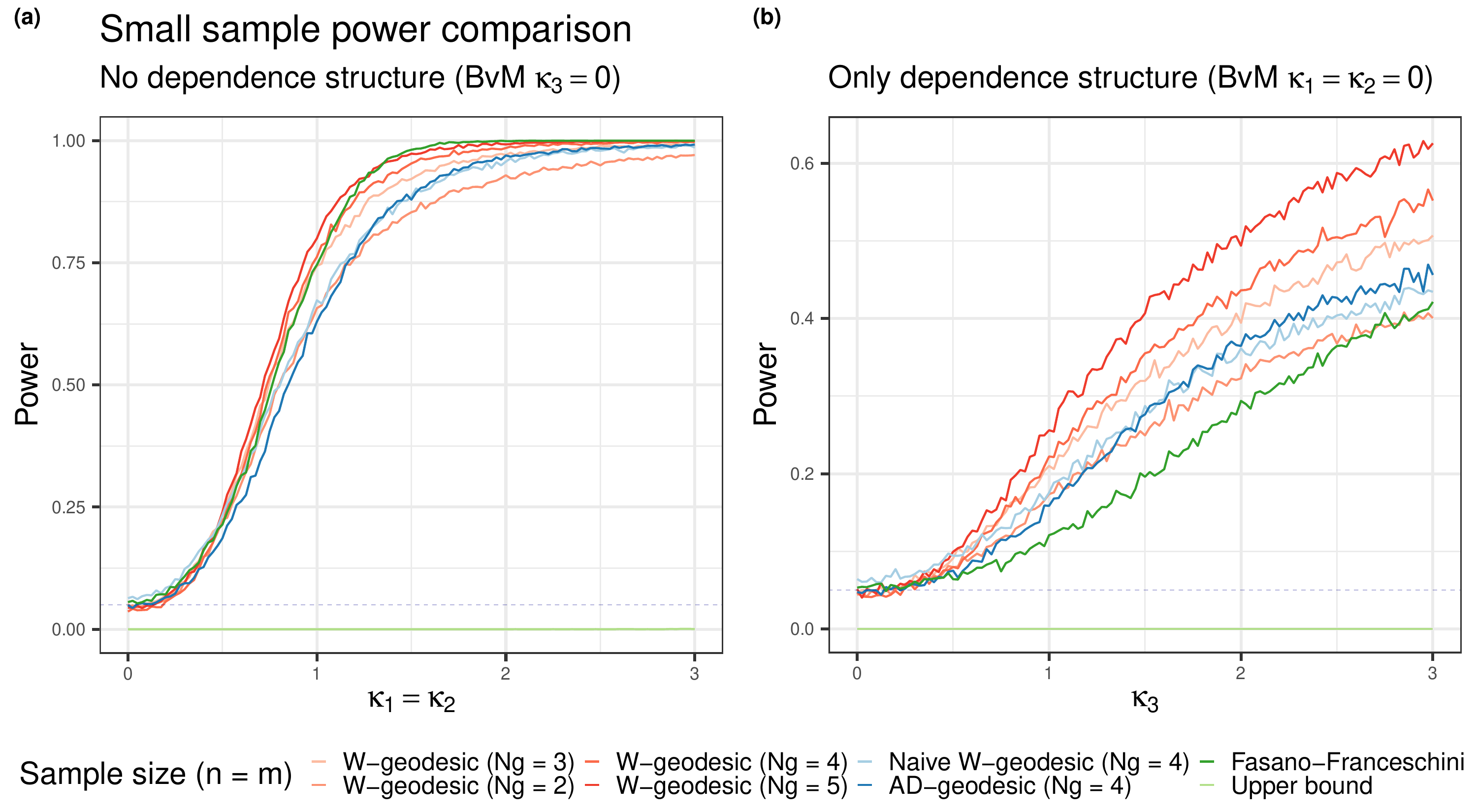}
    \vspace{-5mm}
    \caption{Empirical power of two-sample goodness-of-fit tests for measures supported on $\mathbb{T}^2$, under bivariate von Mises (BvM) alternatives with no dependence structure and different marginal laws (a) and with equal marginal laws and dependence structure (b). The simulated samples had sizes $n=m=50$. The empirical power corresponds to the proportion of rejections at level $\alpha=0.05$ (dashed line) among $5000$ repetitions of the test for fixed concentration parameters.}
    \label{fig:small_sample}
\end{figure}

The first conclusion that we can state after Figure \ref{fig:small_sample} is that \added{the test \eqref{test_epsilon} has zero power for small sample sizes}. This was expected by Proposition \ref{prop:epsilon}, as large values of $n,m$ are required to ensure sharp bounds. However, some interesting conclusions can be extracted regarding the other tests. First, the test \eqref{torus_marg_test} has power $\alpha$ under $H_0$. Indeed, further simulations confirmed that the approach described in Section \ref{sec:ng_geodesics} ensures the uniformity of the combined $p$-value's null distribution. Together with the illustrated consistency of test \eqref{torus_marg_test}, we can observe the considerable gain in power when comparing measures with the Wasserstein statistic \eqref{circle_statistic} by choosing an optimal origin on the circle. The choice of a random origin (`Naive W-geodesic' curve) or the use of techniques that do not rely on Optimal Transport (`AD-geodesic' or Fasano-Franceschini curve) notably reduce the test power, specially when differences are presented on the dependence structure (Figure \ref{fig:small_sample}b). Finally, the choice of the number $N_g$ of geodesics seems to have an effect on power. As one could have expected, increasing the number of geodesic projections improves the test's ability to detect slighter differences. Consequently, the practitioner is entitled to indefinitely increase $N_g$, paying back on computation time (or implementation complexity, if geodesics are randomly chosen, see Appendix \ref{appendix:geodesic_sampling}).

\subsection{Asymptotic performance}\label{section::asymptotic_performance}


This section is devoted to assess the suitability of the upper bound testing technique \eqref{test_epsilon} when large sample sizes are available. Here, we studied the relative efficiency of tests \eqref{test_epsilon} and \eqref{torus_marg_test} for the same converging alternatives as in Section \ref{section::small_sample_performance}, with $n=m\in\lbrace 1000, 1500, 2000\rbrace$. Results are shown in Figure \ref{fig:asymptotic_performance}, where the parameter of interest ($\kappa_1=\kappa_2$ or $\kappa_3$) took values in $\lbrace 0.1, 0.2,\ldots,4\rbrace$, and the empirical power was estimated as the proportion of rejections at level $\alpha=0.05$ among $1000$ repetitions of each test.

\begin{figure}[tp]
    \centering
    \includegraphics[scale = 0.24]{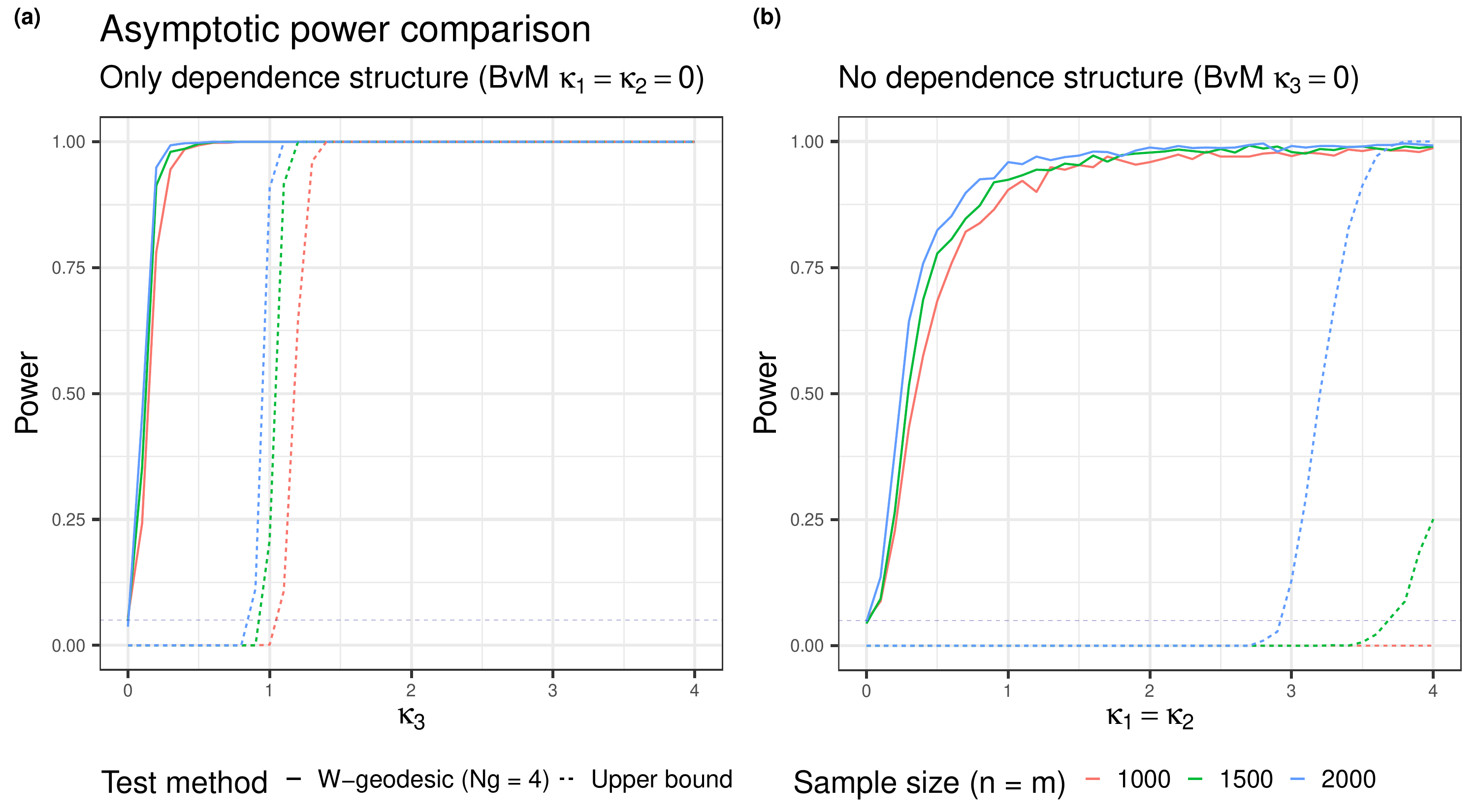}
    \vspace{-5mm}
    \caption{Empirical power of two-sample goodness-of-fit tests for measures supported on $\mathbb{T}^2$, under bivariate von Mises (BvM) alternatives with no dependence structure and different marginal laws (a) and with equal marginal laws and dependence structure (b). The empirical power corresponds to the proportion of rejections at level $\alpha=0.05$ (dashed line) among $1000$ repetitions of each test for fixed concentration parameters.}
    \label{fig:asymptotic_performance}
\end{figure}

Figure \ref{fig:asymptotic_performance} shows that the test \eqref{test_epsilon} is powerful when sample sizes are large enough. As its corresponding $p$-value has been defined as an upper bound of the actual $p$-value \eqref{ideal_pvalue}, it will be quite a conservative test and, therefore, relatively less efficient than \eqref{torus_marg_test}. This is illustrated in both panels. In any case, the test \eqref{test_epsilon} can be useful in practice. Besides the detection of big differences, the practitioner may be interested in the acceptance of small and controlled discrepancies between samples, which may be due, for instance, to experimental inaccuracies. In scenarios where a less conservative method as \eqref{torus_marg_test} may detect such differences, one might prefer to rely on a test method that allows slight dissimilarities and stands out only the more relevant ones. Consequently, even if the test \eqref{test_epsilon} is clearly less efficient than our first candidate \eqref{torus_marg_test}, we believe it can be of interest in some practical scenarios, such as several situations appearing in Structural Biology problems. This is further discussed in Section~\ref{sec:discussion}.

\subsection{Application to protein structure analysis}\label{section::proteins}

A method to accurately compare local structural preferences in conformational ensemble models of proteins is useful to investigate sequence-structure-function relationships, allowing for instance to understand the effect of mutations. The local structure of a protein is determined by two dihedral angles usually denoted by $\phi$ and $\psi$, which describe the conformational state of each amino acid residue along the sequence \cite{Branden:1999,Liljas:2009}. For most amino acid types (for all excepting proline and glycine), the distribution of $\phi$ and $\psi$ angles is supported on the same subset of $\mathbb{T}^2$, which, even if there exist some physically forbidden regions \added{due to strong repulsive forces between non-bonded atoms at short distance}, is connected and has a smooth boundary. We can also assume that density is continuously differentiable and strictly positive in its support, so that Assumption~\ref{Lip} is satisfied. 

The aim of this section is to make use of the tests \eqref{torus_marg_test} and \eqref{test_epsilon} to show that the distribution of $(\phi, \psi)$ does not depend only on the amino acid type, but also on the sequence context, and particularly on the closest neighbors. This corresponds to rejecting Flory’s isolated-pair hypothesis \cite{Flory1969}. Even if the importance of the closest neighbors effect is widely accepted in the Structural Biology community \cite{kabat,Gibrat,Betancourt2004,Rata2010,Ting2010}, only purely descriptive methods have been employed to state so, and no goodness-of-fit techniques have been used to the best of our knowledge. For a given amino acid $C$, we denote by $P_C$ the distribution of $(\phi,\psi)$ supported on $\mathbb{T}^2$. If we take into account the identities $L, R$ of $C$'s left and right neighbors, the distribution of $(\phi,\psi)$ is now given by $P_{LCR}$. The objective is to test
\begin{equation}\label{prot_test}
    H_0: P_C=P_{LCR}\qquad\textrm{against}\qquad H_1: P_C\neq P_{LCR}
\end{equation}
to assess whether nearest neighbors significantly affect dihedral angles distributions. An example of two samples drawn from $P_C$ and $P_{LCR}$ is depicted in Figure \ref{fig:example_flory_ejs}. For the analysis presented here, we used a structural database of three-residue fragments (also called tripeptides) extracted from experimentally-determined high-resolution protein structures~\cite{Estana:19}. The large available sample sizes allow us to illustrate the asymptotic behaviour of \eqref{test_epsilon}. We selected the 71 tripeptides $L$-$C$-$R$ for which the database contained more than $3000$ points. For each one, we compared the corresponding sample of $(\phi,\psi)$ values with an equally-sized sample drawn from $P_C$ (sampled from the sub-database containing $(\phi,\psi)$ values from tripeptides having $C$ as central amino-acid). The data were rescaled to $[0,1]\times[0,1]$ before applying the tests. As the $p$-values for the test \eqref{torus_marg_test} are computed by Monte Carlo simulation, they are lower-bounded by $1/N_{MC}$, where $N_{MC}$ is the number of Monte Carlo replicas \cite{PhipsonSmyth2010}. This point is important here, as due to the large number of performed tests, we had to correct $p$-values for multiplicity \cite{holm}. The results are depicted in Figure \ref{fig:flory_pvalues}, where we show the empirical cumulative distribution function of both tests' corrected $p$-values, for three increasing ranges of sample sizes.

\begin{figure}[tp]
\centering
\includegraphics[scale=0.3]{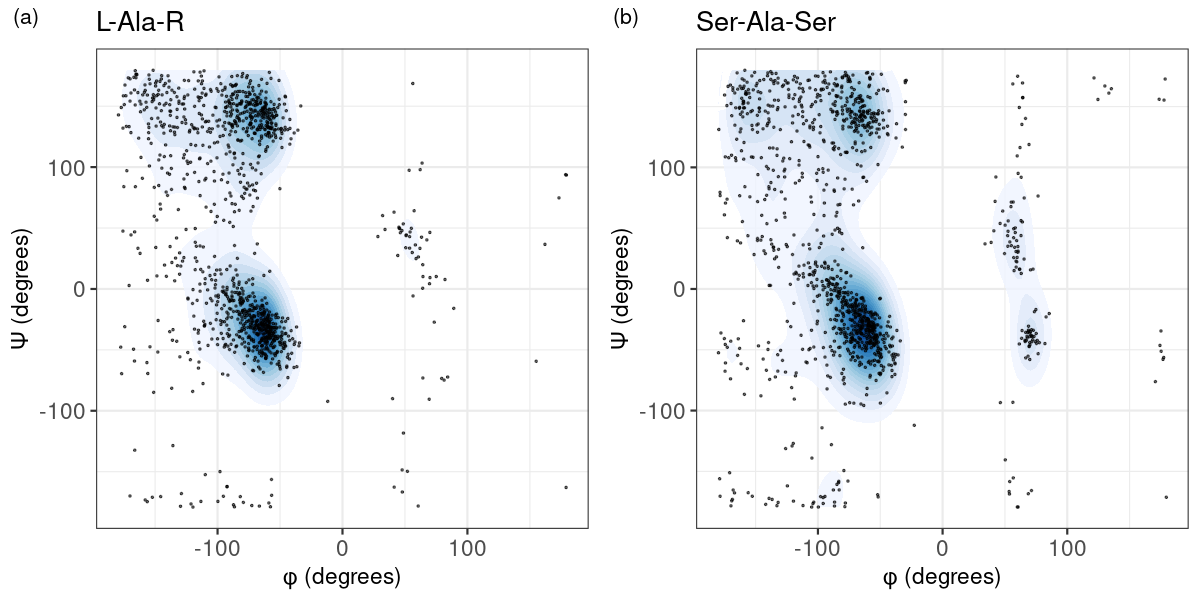}
\caption{(a) Sample and kernel density estimate of alanine $(\phi,\psi)$ distribution when the identity of its left and right neighbor is not taken into account. (b) Sample and kernel density estimate of $(\phi,\psi)$ distribution corresponding to tripeptide Ser-Ala-Ser (a fragment of the three consecutive amino-acids serine, alanine, serine).}
\label{fig:example_flory_ejs}
\end{figure}

From Figure \ref{fig:flory_pvalues}, we can state that the geodesic projection test \eqref{torus_marg_test} strongly rejects the null hypothesis at level $\alpha=0.05$ for the three considered sample size ranges, being all $p$-values truncated to the Monte Carlo precision. \added{Repeating the same analysis for $N_g=3,4$ did not change the shape of the \eqref{torus_marg_test} $p$-values curves, which was expected as higher values of $N_g$ yield a power increase.} A clear asymptotic behaviour is observed for the upper bound technique \eqref{test_epsilon}, as power at level $\alpha$ tends to one when sample sizes increase. Note that, for the largest range of sample sizes, \eqref{test_epsilon} is relatively more efficient than \eqref{torus_marg_test}, due to the Monte Carlo truncation. Both procedures lead to rejection of the null hypothesis, and therefore to the statement that nearest neighbors effect on $(\phi,\psi)$ distributions is statistically significant. This analysis suggests that both \eqref{torus_marg_test} and \eqref{test_epsilon} are suitable for assessing differences on local protein structures, as the available sample sizes (which may be up to $\sim 10^5$ in some practical scenarios) are large enough to state significant conclusions.

\begin{figure}[htp]
    \centering
    \includegraphics[scale=0.4]{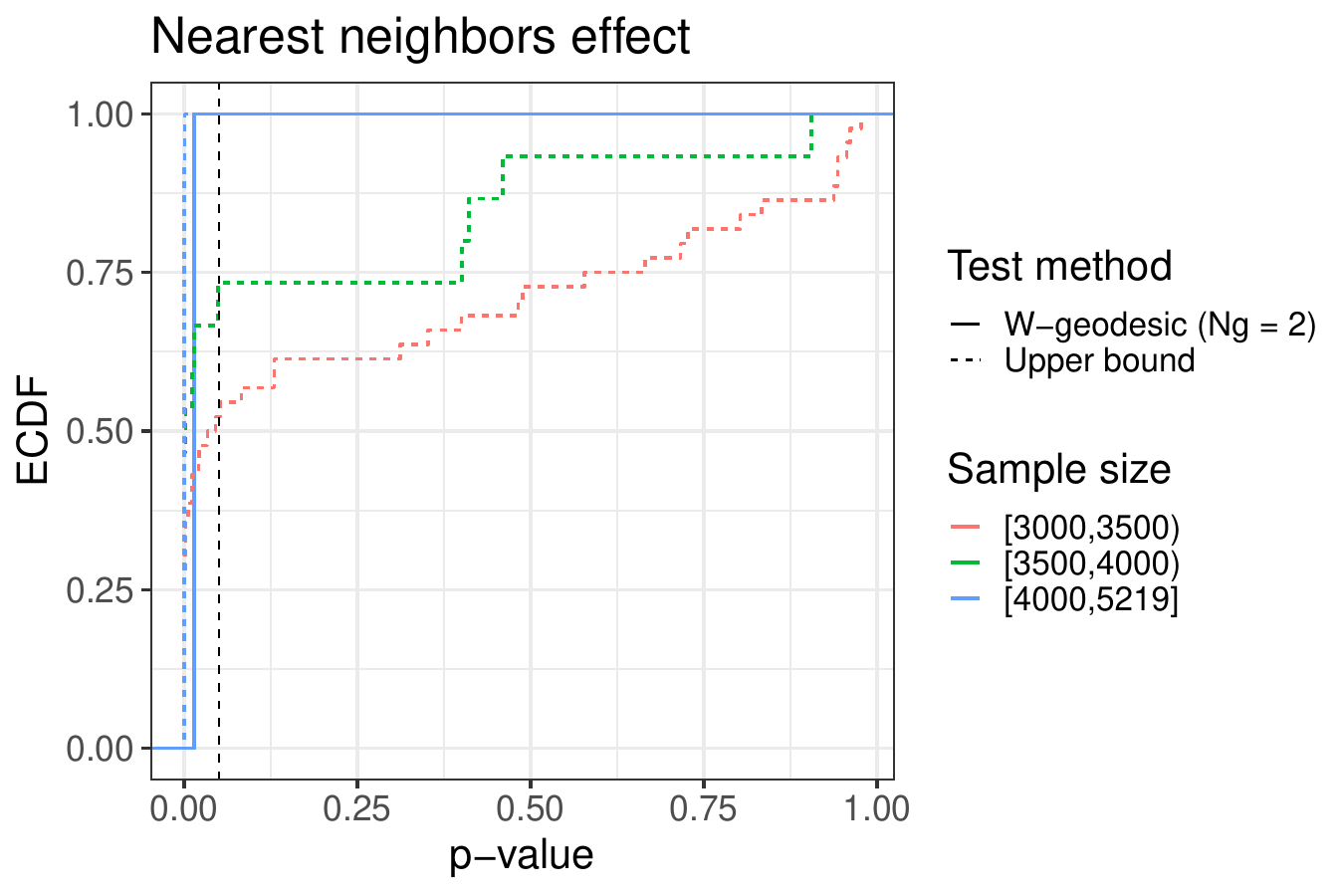}
    \caption{Empirical cumulative distribution function of $p$-values corresponding to test hypotheses \eqref{prot_test} with \eqref{torus_marg_test} (`W-marginal') and \eqref{test_epsilon} (Upper bound) testing methods, for $71$ different combinations of $L$,$C$,$R$. To illustrate the asymptotic behaviour, $p$-values were classified in three ranges of sample sizes. For each test method, $p$-values were corrected for multiplicity using Holm-Bonferroni correction \cite{holm}. Marginal test $p$-values were computed with a Monte Carlo simulation of $N_{MC}=5000$ replicas. The black dashed line indicates an arbitrary significance level of $\alpha=0.05$.}
    \label{fig:flory_pvalues}
\end{figure}

\section{Discussion}\label{sec:discussion}

The main goal of this work was to define suitable two-sample goodness-of-fit tests for measures on $\mathbb{T}^2$. This naturally led us to enrich the existing theoretical results \cite{CORDEROERAUSQUIN1999199, manole2021plugin, McCann} on Optimal Transport for periodic measures. In particular, we studied the shape of the solutions to the Monge problem \eqref{monge}, which allowed the extension of a Central Limit Theorem to $\mathbb{T}^d$, for any $p>1$. Our original inspiration when first investigating these theoretical results was to use the Central Limit Theorem \ref{Theo:centrallim2} to define a two-sample asymptotic test. However, the derived limit distribution degenerates when $P=Q$ and prevents such an application. Nevertheless, the Wasserstein distance on $\mathbb{T}^2$ for the quadratic cost was used to define two efficient testing techniques, which address our initial goals. 

The first approach bypasses the dimension problem by projecting the measures to closed geodesics on $\mathbb{T}^2$ and subsequently test their equality. This required the investigation of how to project samples on closed geodesics and, moreover, how to conveniently sample closed geodesics. The answers we propose here, notably in Sections \ref{appendix:geodesic_sampling} and \ref{appendix:geodesic_projection}, together with their supplied practical implementations, may be of interest in further practical situations. Furthermore, they suggest one possible extension of the Sliced Wasserstein distance \cite{sliced} to the two-dimensional flat torus. As closed geodesics on $\mathbb{T}^2$ are isomorphic to $\mathbb{R}/\mathbb{Z}$, the equality of the projected measures is assessed through a two-sample Wasserstein test on the circle which, to the best of our knowledge, is the only efficient procedure proposed up to now.

The second proposed approach consists in upper-bounding the exact $p$-values \eqref{ideal_pvalue}. This is possible thanks to the derived concentration inequalities (\ref{bound_mean_dev}) for the two-sample empirical Wasserstein distance with the quadratic cost, and to the improved convergence speed of its expectation, as shown in Lemma \ref{Lemma:twosample_exp}. As with any upper-bounding technique, the corresponding test is conservative and only efficient for large sample sizes, which reduces its range of application. However, this test could be relevant in some practical scenarios. For example, Molecular Dynamics simulations (which simulate the temporal evolution of the structure of a protein using force-fields based on physical models), produce samples on $\mathbb{T}^2$ that may present small and meaningless differences when re-running simulations multiple times
with slightly different initial conditions. In such a situation, we expect that the first technique \eqref{torus_marg_test} will reject the equality of their corresponding distributions, while the conservative test \eqref{test_epsilon} will accept differences between independent replicas of the same simulation. Consequently, \eqref{test_epsilon} will only detect more important discrepancies, which are the only ones of interest for practical purposes.

Regarding the practical implementation of both tests, some differences appear with respect to computing time. The main advantage of \eqref{torus_marg_test} is the explicit formulation of Wasserstein distance on one-dimensional spaces, which avoids the use of any Optimal Transport solver. \added{As a result, its time complexity is linear in the sample size}. However, the statistic null-distribution must be simulated with the desired precision, which may slow down the procedure. Note that, in any case, this distribution can be simulated once and be tabulated for any further implementation. The time complexity of \eqref{test_epsilon} exclusively lies on the Optimal Transport solver chosen to compute Wasserstein distance. For very large sample sizes, this might lead to a substantially slower process. 

The issue of two-sample goodness-of-fit testing studied in Section \ref{section:test} remains largely open. Our contribution in this respect is to propose easily implementable goodness-of-fit testing approaches that are built on top of state-of-the-art tools in Optimal Transport. Finding the exact or asymptotic distribution of the Wasserstein statistic in general dimension remains one of the main unsolved problems of the theory of Optimal Transport, preventing the construction of more efficient two-sample goodness-of-fit tests. An asymptotic approach for measures supported on a finite set has been presented in \cite{Sommerfeld2018} and, in the one-dimensional case, \cite{berthet2019} have obtained a CLT under the null $P=Q$ for deviations of $W_p(P_n,Q_n)$ from the true distance $W_p(P,Q)$ (instead of $\mathbb{E}(W_p(P_n,Q_n))$). The results of \cite{berthet2019} are already quite challenging mathematically, and extensions to higher dimensions are clearly beyond the scope of the present work.  Altogether,  we believe that the goodness-of-fit tests defined in this paper constitute a relevant building block for the study of the sequence-structure-function relationship in proteins, and in particular for Intrinsically Disordered Proteins (IDPs), allowing their structural investigation with mathematical guarantees. Furthermore, the interest of the techniques here presented may go beyond the Structural Biology community, as they allow solving the goodness-of-fit testing problem for two distributions lying in general periodic spaces, which appears in various application domains. 

\section*{Code availability}\label{sec:code}

The test approaches presented in this work are implemented in the \textsf{R} package \texttt{torustest}, available at \href{https://github.com/gonzalez-delgado/torustest}{https://github.com/gonzalez-delgado/torustest}, together with the algorithms introduced in Appendix \ref{appendix:geodesics}. Empirical Wasserstein distances were computed using the \textsf{R} package \texttt{transport} \cite{transport}.

\section*{Acknowledgements}
This work was supported by the AI Interdisciplinary Institute ANITI, which is funded by the French ``Investing for the Future – PIA3'' program under the Grant agreement ANR-19-PI3A-0004, and by the ANR LabEx CIMI (grant ANR-11-LABX-0040) within the French State Programme ``Investissements d’Avenir''.

The authors are grateful to the anonymous referees whose comments and suggestions have greatly improved the manuscript. 

\appendix


\section{Geodesics on \texorpdfstring{$\mathbb{T}^2$}{the flat torus}: practical considerations}\label{appendix:geodesics}

This Section is devoted to address some practical questions that arise when defining the test proposed in Section \ref{sec:test_geodesics}. In Appendix \ref{appendix:geodesic_sampling}, we propose a sampling method to prevent the practitioner from explicitly choosing the $N_g$ geodesics, letting them be chosen randomly with respect to a given distribution. In Section \ref{appendix:geodesic_projection}, we propose an algorithm to project a pair of samples on $\mathbb{T}^2$ to a given closed geodesic.

\subsection{Sampling closed geodesics}\label{appendix:geodesic_sampling}

As the closed geodesics on $\mathbb{T}^2$ are given by the canonical projections of straight lines on $\mathbb{R}^2$ with rational slope, sampling from the set of all closed geodesics is equivalent to sampling from $\mathbb{Q}$, which is a countable set. This prevents the sampling to be uniform, in the sense that geodesics can not be equiprobable. Indeed, if $\mathbb{P}(q)=c$ for all $q\in\mathbb{Q}$, by countable additivity $\mathbb{P}(\mathbb{Q})=\sum_{q\in\mathbb{Q}}c$, which is zero if $c=0$ and $\infty$ otherwise. In consequence, as we have to assign different weights to rational slopes, we will opt for \textit{simpler} geodesics to be more probable, in order to ease computational implementations. To achieve so, we can consider the random variable $Q=A/B$, studied in detail in \cite{petroni2019}, where $B$ follows a geometric distribution of parameter $p$ and, for a given denominator $B=b$, $A$ is uniform on $\lbrace 0,1,\ldots,b\rbrace$. Note that $Q$ maps into $\mathbb{Q}\cap[0,1]$. As $p$ increases, $A$ and $B$ take smaller values and the corresponding geodesics revolt less over the torus. Conversely, when $p\rightarrow 0$,
 $\mathbb{P}(Q=q)\rightarrow 0$ for all $q\in\mathbb{Q}\cap[0,1]$, and $Q$ is \textit{asymptotically equiprobable} \cite{petroni2019}. However, small values of $p$ yield extremely high values of $A$ and $B$ and, consequently, unmanageable geodesics with a too-big number of revolutions. The distribution of $Q$ for different values of the parameter $p$ is illustrated in Figure \ref{fig:rational_sampling}. Here, we will ask $p\geq 0.1$ for computational simplicity.
 
 \begin{figure}[tp]
     \centering
     \includegraphics[scale = 0.225]{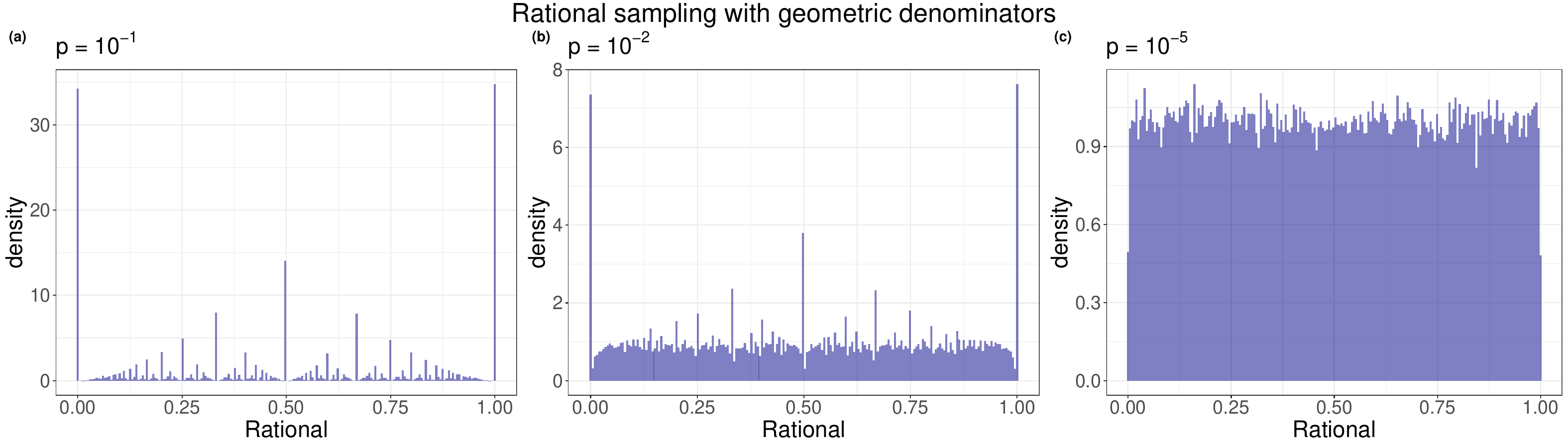}
     \caption{Histograms representing the distribution of the random variable $Q$, for different values of the parameter $p$. For $p=0.1$ (a), rationals with small values of $A$ and $B$ have more weight and, therefore, simpler geodesics are prioritized.}
     \label{fig:rational_sampling}
 \end{figure}
 
 Note that rationals in $\mathbb{Q}\cap[0,1]$ yield to straight lines in $\mathbb{R}^2$ whose director vector $(B,A)$ lies in the first (eq. fifth) octant. To cover all the set of closed geodesics, we uniformly assign an octant to each realization of $Q$ and transform its coordinates appropriately. As we would like all the $N_g$ $p$-values to be independent, we must only accept samples with $N_g$ different geodesics. This may be a problem if $N_g$ is too big, and might require decreasing the value of $p$. Nevertheless, for a small number ($\lesssim 30)$ of geodesics we can keep $p\sim 0.1$ and easily get samples with no repetitions. If one needs to perform the test for large values of $N_g$, we recommend to explicitly choose geodesics a priori to avoid this problem, leaving the sampling method for controlled values of $N_g$. The complete sampling procedure is described in Algorithm \ref{algorithm_rational}, which takes $N_g$ and $p$ as arguments and retrieves $N_g$ director vectors. \added{In Algorithm \ref{algorithm_rational}, $\mathcal{M}_{N_g\times 2}(\mathbb{Z})$ denotes the set of $(N_g\times 2)$-matrices with entire entries and we define $\mathring{\gcd}$ as
 \begin{equation*}
     \mathring{\gcd}(b,a)=\left\{ \begin{array}{lcc}
             \gcd(b,a) &  \textrm{if} & a\neq 0,\\
             b &   \textrm{otherwise},
             \end{array}\right.
 \end{equation*}
for $a,b\in\mathbb{Z}$ with $b\neq 0$.}

\begin{algorithm}
\caption{Geodesics sampling}\label{algorithm_rational}
\begin{algorithmic}
\Require $N_g\in\mathbb{N},\,p=0.1$
\Ensure $G \in\mathcal{M}_{N_g\times 2}(\mathbb{Z})$
\State $G \gets 0\in\mathcal{M}_{N_g\times 2}(\mathbb{Z})$
\While{$|\lbrace i=1,\ldots,N_g\,:\, G_{ik}=G_{jk}\,\,\forall\, k\in\lbrace 1,2\rbrace\quad \textrm{for any }j\in \lbrace1,\ldots,N_g\rbrace\backslash\lbrace i\rbrace\,\rbrace|>0$}
\For{$i \gets 1$ to $N_g$}
\State $b \gets \mathcal{G}(p)$
\State $a \gets \mathcal{U}(\lbrace 0,1,\ldots,b\rbrace)$
\State $u \gets (b,a)/\added{\mathring{gcd}(b,a)}$ \Comment{Director vector in $\mathbb{R}^2$}.

\State $o \gets \mathcal{U}(\lbrace 1,2,3,4 \rbrace)$ \Comment{Octant of the upper semi-circle}.
\If{$o=2$}
\State $u\gets (a,b)/\added{\mathring{gcd}(b,a)}$
\ElsIf{$o=3$}
\State $u\gets (-b,a)/\added{\mathring{gcd}(b,a)}$
\ElsIf{$o=4$}
\State $u\gets (-a,b)/\added{\mathring{gcd}(b,a)}$
\EndIf
\State $G_{i\cdot}\gets u$
\EndFor
\EndWhile
\end{algorithmic}
\end{algorithm}

\subsection{Projection to a closed geodesic}\label{appendix:geodesic_projection}

Let $a,b\in\mathbb{Z}$, with $b\neq 0$, and $u=(a,b)$ the director vector of a straight line $r_u^0$ containing the origin $(0,0)$. Let $I_a$ be the real interval $(\min(a,0),\max(a,0))$, being $I_b$ analogously defined. We aim to project a pair of samples into the geodesic given by the canonical projection of $r_u^0$. To do so, we first consider the finite set $\mathcal{P}_u$ of the points in $I_a\times I_b$ where $r_u^0$ cuts the lines $x=z_a$, $y=z_b$ for $z_a\in I_a\cap\mathbb{Z}$ and $z_b\in I_b\cap\mathbb{Z}$:
\begin{equation*}
    \mathcal{P}_u=\lbrace (x,z)\,:x\in I_a, z\in\mathbb{Z}\rbrace\cap\lbrace (z,y)\,:y\in I_b, z\in\mathbb{Z}\rbrace\cap r_u^0.
\end{equation*}
An example is presented in Figure \ref{fig:point_set}a. Then, we consider the set $\mathcal{L}_u$ of straight lines of director vector $u$ and containing the points of $\mathcal{P}_u\cup \lbrace (0,1),(1,0),(0,0)\rbrace$ transferred to $[0,1]\times[0,1]$ by subtracting the integer part of its coordinates. If we denote $r_v^p$ the straight line containing $p=(p_x,p_y)\in\mathbb{R}^2$ and having $v$ as director vector, we can write $\mathcal{L}_u$ as follows
\begin{equation*}
    \mathcal{L}_u=\lbrace r_u^q\,:\,q=(p_x-\left[p_x\right],p_y-\left[p_y\right])\,,\,p\in\mathcal{P}_u\cup\lbrace (0,1),(1,0),(0,0)\rbrace \rbrace.
\end{equation*}
This is illustrated in Figure \ref{fig:point_set}b. In a first step, each point in $[0,1]\times[0,1]$ will be projected to the closest straight line in $\mathcal{L}_u$. Then, projections $(x_u,y_u)$ outside $[0,1]\times[0,1]$ will be replaced by the elements $(x_u',y_u')\in[0,1]\times[0,1]$ such that $(x_u,y_u)\mathcal{R}(x_u',y_u')$, where $\mathcal{R}$ is the one defined in the begging of Section \ref{first_section}. These two steps are depicted in Figure \ref{fig:sample_projections}.

\begin{figure}[h!]
    \centering
    \includegraphics[scale=0.27]{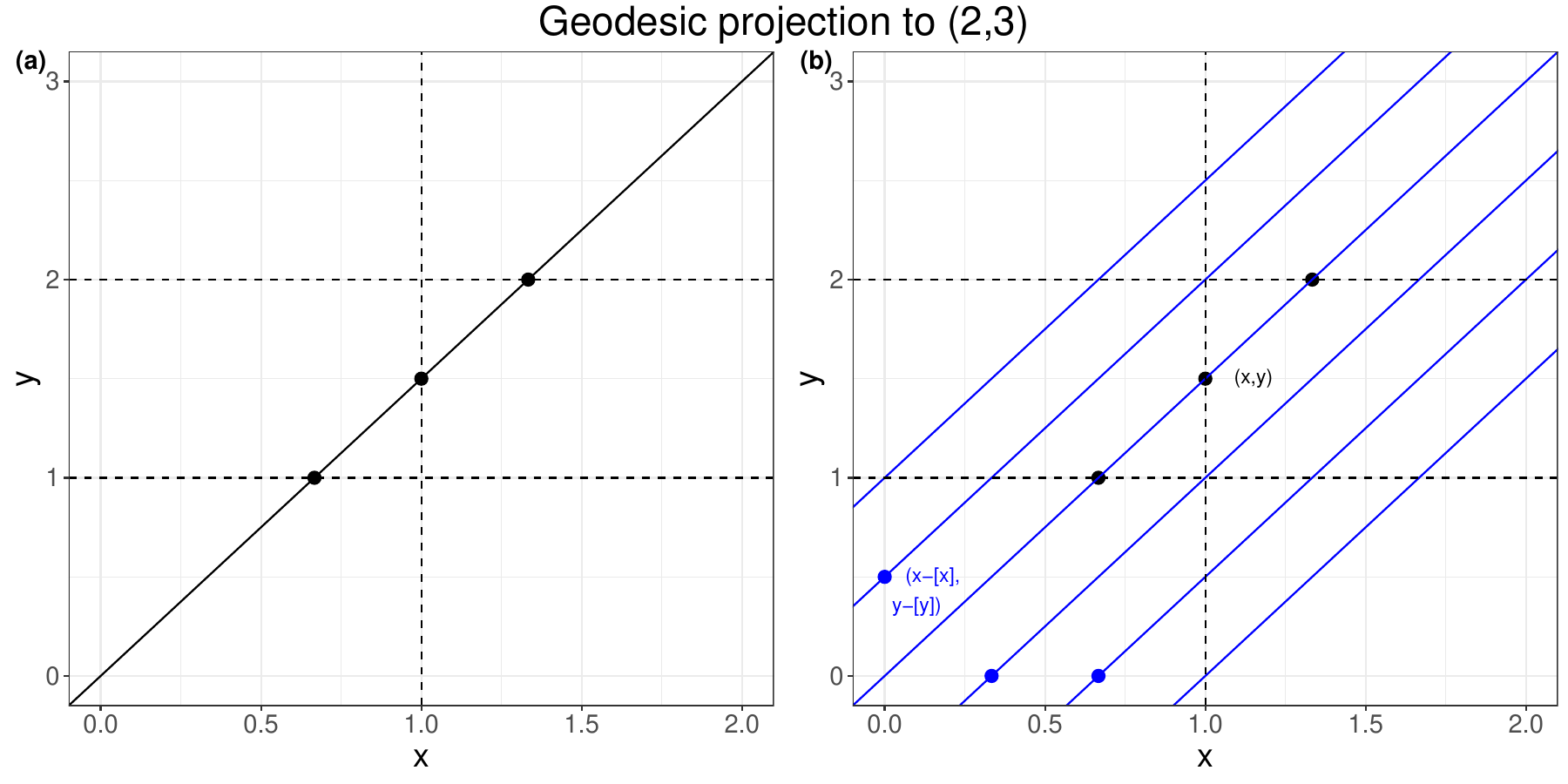}
    \caption{First steps of the projection algorithm for the closed geodesic corresponding to the straight line of director vector $u=(2,3)$. The three points in black constitute the ensemble $\mathcal{P}_u$. In (b), points of $\mathcal{P}_u$ are transferred to $[0,1]\times[0,1]$ by subtracting to their coordinates their integer parts. The blue lines are the elements of $\mathcal{L}_u$.}
    \label{fig:point_set}
\end{figure}

The last step is to relocate all the projections on $\mathbb{R}/\mathbb{Z}$. To do so, we put the segments $\mathcal{L}_u\cap([0,1]\times[0,1])$ in order, following the spiral path. This corresponds to transfer back the points to the straight line $r_u^0$ of Figure \ref{fig:point_set}a. Let $(x_u,y_u)\in r_u^p\in\mathcal{L}_u$. The element $t_u\in\mathbb{R}/\mathbb{Z}$ will be parameterized as 
\begin{equation*}
    t_u=\frac{\|\tilde{p}\|+\|(x_u,y_u)\|}{\|u\|}\in[0,1),
\end{equation*}
where $\tilde{p}\in\mathcal{P}_u\cup\lbrace(0,0)\rbrace$ is the one such that $p_x=\tilde{p}_x-\left[\tilde{p}_x\right]$ and $p_y=\tilde{p}_y-\left[\tilde{p}_y\right]$.

\begin{figure}[h!]
    \centering
    \includegraphics[scale=0.25]{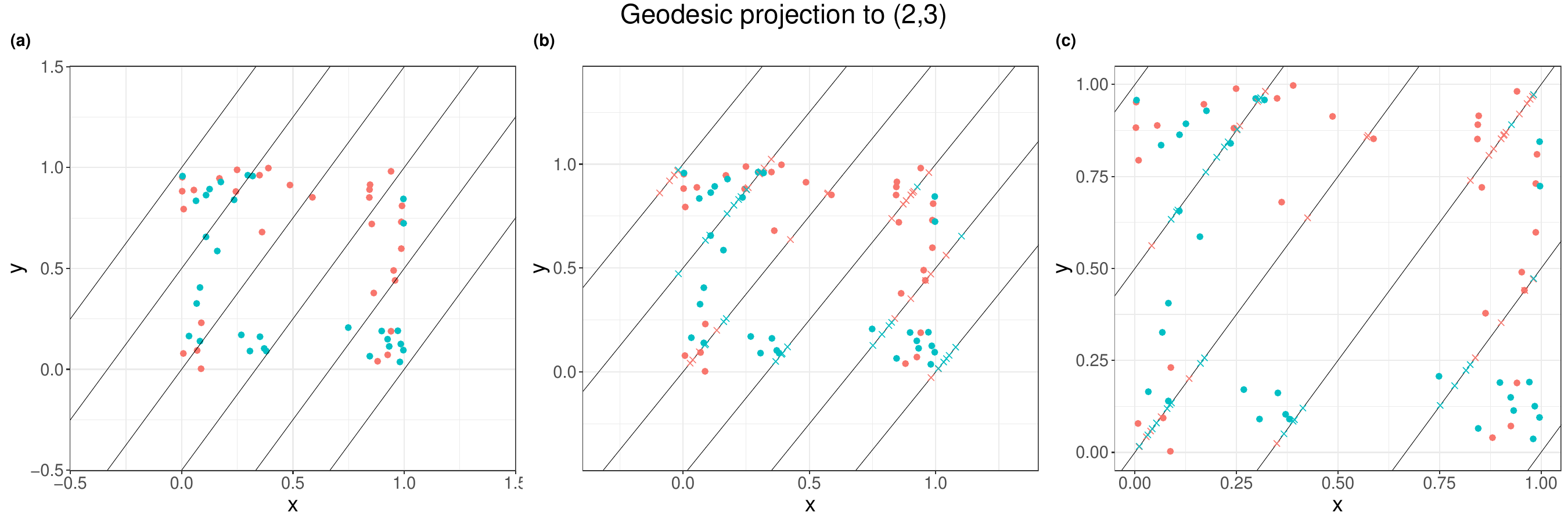}
    \caption{Projection to the closed geodesic given by the director vector $u=(2,3)$ of a pair of samples of size $n=m=30$ drawn from a uniform distribution on $\mathbb{T}^2$. Black lines are the elements of $\mathcal{L}_u$. In $(a)$, the given samples distinguished by colors. In $(b)$, their projections to the closest line in $\mathcal{L}_u$ are represented by colored crosses. In $(c)$, projections outside $[0,1]\times[0,1]$ are relocated in $[0,1]\times[0,1]$ according to the equivalence relation $\mathcal{R}$.}
    \label{fig:sample_projections}
\end{figure}

\section{Proofs}\label{app}

\subsection{Proofs of Section \ref{section_measures}}

\begin{proof}[Proof of Theorem \ref{TheoremExistenceOfArrangements}] Recall that we denote the interior of the support of a  measure $\mu$ (over $\mathbb{T}^d$ or $\R^d$)  as $\mathcal{X}_{\mu}$. Since $\mathbb{T}^d$ is a Polish space, Theorem 4.1 in \cite{Villani2008} implies that there exists a solution $\pi^*$ of \eqref{kant}. Additionally, Theorem 5.10 in \cite{Villani2008} establishes that  $\text{supp}(\pi^*)$ is $d^p$-cyclically monotone. 
More precisely, by Theorem 5.10 in \cite{Villani2008}, this support lies on the graph of the $d^p$-differential
$$\partial^{d^p}f(\bx)=\{\by:\ f(\bz)\leq f(\bx)+ d^p( \bz,\by)-d^p( \bx,\by), \text{ for all $\bz\in\mathbb{T}^d$}\}$$
of a function $f$ solving  \eqref{dual}. Its graph is denoted by $\partial^{d^p}f=\{(\bx,\by):\ \by\in \partial^{d^p}f(\bx)\}$. These definitions of $d^p$-differential  and $d^p$-concave functions apply verbatim to $\|\cdot\|^p $-differential  and $\|\cdot\|^p$-concave functions with the obvious notation. 
Let $\Gamma$ be the set defined in \eqref{Gamma}, $ \{(\x_k+\p_k, \y_{k}+\p_k) \}_{k=1}^n\subset \Gamma$ be a sequence and  ${\sigma:\{ 1, \dots, n \}\rightarrow \{ 1, \dots, n \}}$ be a bijection. Then, the definition of $\Gamma$ implies that
\begin{align*}
\sum_{k=1}^n \|\x_k-\y_k  \|^p & =\sum_{k=1}^n d^p(\bx_k, \by_k) \\
& \leq \sum_{k=1}^n d^p(\bx_k, \by_{\sigma(k)}) \\
& \leq \sum_{k=1}^n \|\x_k+\p_k-{\y}_{\sigma(k)}- \p_{\sigma(k)} \|^p,
\end{align*}
which means that $\Gamma$  is $\| \cdot \|^p-$cyclically monotone. Therefore, $\Gamma\subset \partial^{\| \cdot\|^p}\varphi_p$, for some $\| \cdot \|^p$-concave function $\varphi_p$. Now, recall from Theorem 3.3 and Proposition 3.4 in \cite{GangboMccann}, that 
\begin{enumerate}
    \item  The set of differentiablity $$\operatorname{dom}(\nabla\varphi_p)=\left\lbrace \x\in \R^d: \  \partial^{\| \cdot\|^p}\varphi_p= \left\lbrace \x-\left(\frac{1}{p}\|\nabla\varphi_p(\x) \|\right)^{\frac{2-p}{p-1}}\nabla\varphi_p(\x)\right\rbrace\right\rbrace $$ has full  Lebesgue measure in $\operatorname{dom}(\varphi_p)=\{\x\in \R^d:\ \varphi_p(\x)\in \R\}\supset \mathcal{X}_{\mu_P}$,
    \item  The relation 
    $\mathbf{S}_p(\x) = \x-\left(\frac{1}{p}\|\nabla\varphi_p(\x) \|\right)^{\frac{2-p}{p-1}}\nabla\varphi_p(\x)$ defines a Borel function in $\operatorname{dom}(\nabla\varphi_p)$, and
    \item  The equality $\{\mathbf{S}_p(\x)\}=\{\y: (\x,\y)\in \Gamma\}$ holds for all $\x\in \operatorname{dom}(\nabla\varphi_p)$.
\end{enumerate}
Since $\Gamma\subset  \partial^{\| \cdot\|^p}\varphi_p$, this means that, for all $\x\in \operatorname{dom}(\nabla\varphi_p)$, there exists an unique $\y_{\x}=\mathbf{S}_p(\x) $ such that $(\x,\y_{\x})\in \Gamma$. We observe that, due to the fact that $\mu_P\ll\ell_d$, the measure $\gamma^*=( \mathbf{Id}\times \mathbf{S}_p)\# \mu_P$ on $\R^d\times\R^d$ is well defined, its support is $\| \cdot \|^p$-cyclically monotone and its first marginal is $\mu_P$. We claim that the second marginal is $\mu_Q$. Let $(\bx, \by)\in \pi^*$ be such that $\x+\p\in \operatorname{dom}(\nabla\varphi_p)$, for all $\p\in \Z^d$. Then, for any representative pair, call it  $(\x,\y)\in \R^d$, there exist $\p,\p'\in \Z^d$ such that $(\x+\p,\y+\p')\in \Gamma$. Since
\begin{align*}
     \{\mathbf{S}_p(\x)+\p\}&=\{ \y+\p: \ (\x,\y)\in \Gamma \}\\
     &=\{ \y: \ (\x+\p,\y)\in \Gamma \}\\
     &=\{\mathbf{S}_p(\x+\p)\}=\{\y+\p'\},
\end{align*}
the relation $\by=\overline{\mathbf{S}_p(\bx)}$ holds. Since $\x+\p\in \operatorname{dom}(\nabla\varphi_p)$, for all $\p\in \Z^d$, which is the intersection of sets of full $\mu_P$-measure, the relation $\by=\overline{\mathbf{S}_p(\bx)}$ happens $\mu_P-$a.e. This means that $\pi^*=( \mathbf{Id}\times \overline{\mathbf{S}_p(\bar{\cdot})})\# P$, which proves automatically the claim. Consequently, the existence is proven.

The uniqueness follows from the proof of Corollary 2.4. in \cite{GangboMccann}. Indeed, we can define the set
\begin{equation*}
    S=\bigcup_{\pi^*\text{ solving \eqref{kant}}}\Gamma(\pi^*),
\end{equation*}
where $\Gamma(\pi^*)$ is defined as in \eqref{Gammap} for each $\pi^*$ solving \eqref{kant}. Therefore, taking any finite sequence $ \{(\x_k+\p_k, \y_{k}+\p_k) \}_{k=1}^n\subset S$, there exists at most $n$ different probability measures $\pi_k$, for $k=1, \dots, n$, such that $(\x_k+\p_k, \y_{k}+\p_k) \in \Gamma(\pi_k)$. As all of them are solutions of \eqref{kant} we have, due to the linearity of the optimization in  \eqref{kant} and the convexity of the set $\Gamma(P,Q)$, that the mean $\pi_{0}=\frac{1}{n}\sum_{k=1}^n \pi_k$ is also a solution. Then, its support is contained in a $d^p$-cyclically monotone set, and $\Gamma(\pi_{0})$ is $\|\cdot \|^p$-cyclically monotone, since it contains the sequence  $ \{(\x_k+\p_k, \y_{k}+\p_k) \}_{k=1}^n\subset S$. Consequently, $S$ is $\|\cdot \|^p$-cyclically monotone.

To conclude, repeating the previous arguments, there exists a $\|\cdot \|^p$-concave function $f^S$ such that $S\subset \partial^d f^S $. Moreover, for any other $\varphi_p$, defined as before, it holds that $\partial^d\varphi_p\subset \partial^d f^S $. Then, the equality $$ \x-\left(\frac{1}{p}\|\nabla\varphi_p(\x) \|\right)^{\frac{2-p}{p-1}}\nabla\varphi_p(\x)=\x-\left(\frac{1}{p}\|\nabla f^S(\x) \|\right)^{\frac{2-p}{p-1}}\nabla f^S(\x)$$ holds $\mu_P$-a.e. This proves the uniqueness of $\mathbf{S}_p$ and, consequently, the one of $\mathbf{T}_p$
 
\end{proof}

\begin{proof}[Proof of Theorem \ref{Corollary:equals}] We set $(\x,\y)\in \Gamma$ and observe that $d(\bx, \by)=\| \x-\y\|$. Since $(\bx, \by)\in \text{supp}(\pi^*)$,  Theorem 5.10 in \cite{Villani2008} establishes that if $(f,g)$ solves  \eqref{dual}, it holds
\begin{align*}
f(\bx)=\inf_{\by\in\mathbb{T}^d}\{d(\bx,\by)^p-g(\by)\}.
\end{align*}
Since, for each $(\x,\y)$, there exists $\p\in \Z^d$ such that $d(\bx, \by)=\| \x-\y-\p\|$, we can directly replace $\y$ by $\y+\p$ in the infimum without altering any of the terms. This yields the equality
\begin{align*}
f(\by)=\inf_{z,\y\in \in\R^d \ \bz=\bx}\{\|\z-\y\|^p-g(\by)\},
\end{align*}
and allows to define the following periodic $\| \cdot\|^d$-concave function in $\R^d$: 
\begin{align*}
\hat{\varphi}_p(\x)=\inf_{\y\in \in\R^d}\{\|\x-\y\|^p-g(\by)\}=f(\bx).
\end{align*}
We claim that $\nabla\hat{\varphi}_p=\nabla\varphi_p$ for $\mu_P$-a.e., which implies the equality of both $ \hat{\varphi}_p$ and $ \varphi_p$, in each connected component of $\operatorname{supp}(\mu_P)$. By assumption,  $\operatorname{supp}(\mu_P)=\bigcup_{p\in \Z^d} p+A$ is a union of connected sets. By periodicity we can restrict our study to the connected set $A$, where the claim yields $\nabla\hat{\varphi}_p=\nabla\varphi_p$ for $\ell_d$-a.e. We can apply Theorem 2.6 in \cite{Alberto} to conclude that $\varphi_p=\hat{\varphi}_p+C$ in $A$, thus in $\operatorname{supp}(\mu_P).$ We prove now the claim. Let $\pi^*$ be a measure solving \eqref{kant}, we know (from Theorem 5.10 in \cite{Villani2008}) that its support lies in the graph of $\partial^{d^p}f$. Therefore, we can define the following $\|\cdot\|^p$-cyclically monotone set (note that this is true by repeating the same arguments as for $\Gamma$):
\begin{align}
\label{Gammap}
\notag
\Gamma(\partial^{d^p}f) & = \{ (\x+p,\y+p):\ \\ 
& ( \bx, \by)\in\partial^{d^p}f,\ \x\in[0,1]^d, \ d(\bx, \by)=\| \x-\y\|\ \text{and}\ p\in\Z^d\},
\end{align}
which satisfies the relation $\Gamma(\pi^*)\subset\Gamma(\partial^{d^p}f)$, with the notation of the proof of Theorem~\ref{TheoremExistenceOfArrangements}. Recall that the relation $ \Gamma(\pi^*)\subset \partial^{\|\cdot\|^p}\varphi_p$ also holds. Moreover, by definition we have $\Gamma(\partial^{d^p}f)\subset \partial^{\|\cdot\|^p}\hat{\varphi}_p$. Since $\mu_P $-a.e. the sets $\partial^{\|\cdot\|^p}\hat{\varphi}_p(\x)$ and $\partial^{\|\cdot\|^p}{\varphi}_p(\x)$ are singletons, and, for $\mu_P $-a.e. $\x$, there exists at least one $\y\in \R^d$ such that $(\x,\y)\in  \Gamma(\pi^*)$, then $\partial^{\|\cdot\|^p}\hat{\varphi}_p(\x)=\partial^{\|\cdot\|^p}{\varphi}_p(\x)$. This implies that the functions $\mathbf{S}_p(\x) = \x-\left(\frac{1}{p}\|\nabla\varphi_p(\x) \|\right)^{\frac{2-p}{p-1}}\nabla\varphi_p(\x)$ and $\hat{\mathbf{S}}_p(\x) = \x-\left(\frac{1}{p}\|\nabla\hat{\varphi}_p(\x) \|\right)^{\frac{2-p}{p-1}}\hat{\varphi}_p(\x)$ are equal $\mu_P $-a.e., which proves the claim.\added{ Note that, under continuity of the optimal transport potential,  their uniqueness only need to be fulfilled $\mu_P $-a.e. } 
\end{proof}

\begin{proof}[Proof of Lemma \ref{Lemma:lip}]
\added{Set $\bx,\bz\in \text{dom}(f)$. Then, by definition 
\begin{align*}
     |f(\bx)-f(\bz)|&=|\inf_{\by\in\mathbb{T}^d}\{d^p(\bx,\by)-g(\by)\}-\inf_{\by\in\mathbb{T}^d}\{d^p(\bz,\by)-g(\by)\}|\\
     &=|\inf_{\by\in\mathbb{T}^d}\{d^p(\bx,\by)-g(\by)\}+\sup_{\by\in\mathbb{T}^d}\{-d^p(\bz,\by)+g(\by)\}|\\
     &\leq \sup_{\by\in\mathbb{T}^d}|d^p(\bx,\by)-d^p(\bz,\by)\}|.
\end{align*}
The mean value theorem yields the inequality $a^p-b^p\leq p|a-b|(a^{p-1}+b^{p-1})$, which holds for any  $a,b\geq 0$. Then, the triangle inequality for $d$ leads to
\begin{multline*}
    |f(\bz)-f(\bx)| \leq p \, d( \bz,\bx) \sup_{\by\in\mathbb{T}^d}|d^{p-1}(\bx,\by)+d^{p-1}(\bz,\by)\}\\
 \leq 2\, p\, d( \bz,\bx) \sup_{ \bz, \bx\in\mathbb{T}^d} \left( d^{p-1}( \bz,\bx) \right)\leq 2 \, p\, d^{\frac{p-1}{2}}\, d( \bz,\bx)
\end{multline*}
where the $d^{\frac{p-1}{2}}$ term comes from the trivial bound of the diameter of $\mathbb{T}^d$. This concludes the proof.} \end{proof}


\begin{proof}[Proof of Theorem \ref{Theoremeq}] Set $\bar{p}\in \mathcal{X}_P$ and assume that $ f_p(\bar{p})=0  $.  Set $\epsilon_m\rightarrow 0$  and consider the sequence of balls $\mathbb{B}_{\epsilon_m}(\bar{p})\subset \text{supp}(P)$,  centered at $\bar{p}$ with radius $\epsilon_n$. Since the ball is a continuity set of $P$, after Portmanteau's Theorem, $P_n\xrightarrow{w} P$ implies that for each $m$ there exists a $n_m$ such that $P_n$ gives mass to $\mathbb{B}_{\epsilon_m}(\bar{p})$ for all $n\geq m_n$. Then, we can extract a sequence $\bar{p}_n\rightarrow \bar{p}$ such that $\bar{p}_n\in\mathcal{X}_{P_n}$. As a consequence, we have that $f_n(\bar{p}_n)\in \R$, and we can set $a_n=-f_n(\bar{p}_n)$ and define $h_n=f_n+a_n$. Recall from Lemma~\ref{Lemma:lip} that all such functions are $L$-Lipschitz in their respective domains. Kirszbraun's Theorem (Theorem~B in \cite{Lang1997KirszbraunsTA}) implies that, without loss of generality, we can consider that $h_n$ (resp.  $f_p$) are $2p$-Lipschitz functions defined in the whole $\mathbb{T}^d$.  The previous reasoning implies that $\{h_n\}_{n\in \N}$ is point-wise bounded for the compact sequence $\{\bar{p}_n\}_{n\in \N}$.  Since all such functions are $2p$-Lipschitz, then Arzelá-Ascoli's Theorem concludes that every subsequence $\{h_{n_k}\}_{k\in \N}$ admits a convergent subsequence $\{h_{n_{k_j}}\}_{j\in \N}$.
\added{Let $h$ be one of those limits. Note that the $d^p-$conjugation is continuous in the sense that 
\begin{multline*}
     |h_n^{d^p}(\bx)-h^{d^p}(\bx)|=|\inf_{\by}\{d^p(\by,\bx)-h_n(\bx)\}-\inf_{\by}\{d^p(\by,\bx)-h(\bx)\}|\\
     \leq \sup_{\by}\{  h_n(\bx)-h(\bx)\}=\| h_n-h\|_{\infty},
\end{multline*}
 for all $\by\in \mathbb{T}^d$. By assumption, we have
$$
    A_n= \int  h_n d\alpha_n+ \int  h_n^{d^p} d\beta_n
    - \int  f_pd\alpha-\int  f_p^{d^p} d\beta\to 0,
$$
and 
\begin{multline*}
    \int  h d\alpha+ \int  h^{d^p} d\beta=\\
   \int  h d(\alpha-\alpha_n)+ \int  h^{d^p} d(\beta-\beta_n)+ \int  (h_n-h )d\alpha_n+ \int  (h_n^{d^p}-h^{d^p}) d\beta_n.
\end{multline*}
Then, the inequality $| \int  (h_n-h )d\alpha_n|\leq  \|h_n-h \|_{\infty}$ gives $
 \int  h d\alpha+ \int  h^{d^p} d\beta=0$. The function $h$ is thus an optimal transport potential.  The uniqueness described in Theorem~\ref{Corollary:equals} and the fact that $\bar{p}_n\rightarrow \bar{p}$ and $h_n(\bar{p}_n)=f_p(\bar{p})=0$ conclude that $f_p$ is the unique possible limit of such subsequences in $\operatorname{dom}(f_p)$.}

\end{proof}

\begin{proof}[Proof of Theorem \ref{Theo:centrallim2}] 

Note that as Theorem~\ref{Theoremeq} holds, since probability measures are supported in a compact set, the torus, then the reasoning of \cite{DelBarrio2019} can be imitated. Here the main steps of the proof for the one-sample case are given. For further details about the proof we refer to the original text.

Efron-Stein inequality, see Chapter 3.1 in~\cite{Boucheron}, states that if $(X'_1,\dots, X'_n) $ is an independent copy of $(X_1,\dots, X_n )$,  then we have the bound
\begin{align*}
\text{Var}(f(X_1,\dots, X_n )) & \leq  \sum_{i=1}^n \mathbb{E}(f(X_1,\dots, X_n )\\ 
& -f(X_1,\dots,X_{i-1}, X'_{i}, X_{i+1},\dots, X_n ).)_+^2.
\end{align*} 
Moreover, if $X_1,\dots, X_n $ are i.i.d,  such inequality can be written as
$$ \text{Var}(f(X_1,\dots, X_n ))\leq n \mathbb{E}(f(X_1,\dots, X_n )-f(X'_1,\dots, X_n ))_+^2.$$
Set the empirical measures $P_n=\frac{1}{n}\sum_{k=1}^n\delta_{X_k}$ and $P'_n=\frac{1}{n}(\delta_{X'_1}+\sum_{k=2}^n\delta_{X_k})$,  and
the values $R_n=\mathcal{T}_p(P_n,Q)-\int f_p dP_n$  and $R'_n=\mathcal{T}_p(P'_n,Q)-\int f_p dP'_n$. Let $f_n$ and $f_n'$ be solutions of the dual problem \eqref{dual} of $\mathcal{T}_p(P_n,Q)$ and $ \mathcal{T}_p(P'_n,Q) $ respectively.  Then from \eqref{dual} we derive that
\begin{align*}
(R_n-R'_n)_+\leq &\frac{1}{n}|f_n(X_1)-f_p(X_1)-f_n(X'_1)+f_p(X'_1)|\\
&+|f_n'(X_1)-f_p(X_1)-f_n'(X'_1)+f_p(X_1')|,
\end{align*}
which together with Theorem~\ref{Theoremeq} yields
\begin{align*}
n(R_n-R'_n)_+\xrightarrow{a.s.}0.
\end{align*}
Since the probability measures are supported in the torus, which is compact, then
$n^2\mathbb{E}(R_n-R'_n)_+^2\rightarrow 0.
$ Finally, we conclude by the so called Efron-Stein's inequality.

\end{proof}

\subsection{Proofs of Section \ref{section:test}} \label{proofs_3}

\begin{proof}[Proof of Lemma \ref{lemma_uniform}] If $G$ is the distribution function of the uniform distribution on $\mathbb{R}/\mathbb{Z}$, we have that
\begin{equation}\label{g_alpha}
    (G-\alpha)^{-1}(t)=\inf\lbrace s\,:\, s>t+\alpha\rbrace=t+\alpha.
\end{equation}
Plugging \eqref{g_alpha} in \eqref{w_circle}, we have
\begin{equation}
    \mathcal{T}_2(P^c,U)=\underset{\alpha\in\mathbb{R}}{\inf}\int_0^1\left(F^{-1}(t)-t-\alpha\right)^2\,dt,
\end{equation}
where the optimal value for $\alpha$ can be found by analytically minimizing the function
\begin{align*}
    H(\alpha)=\int_0^1\left(F^{-1}(t)-t-\alpha\right)^2\,dt=\int_0^1 \left(F^{-1}(t)-t\right)^2\,dt+\\\alpha^2-2\alpha\int_0^1\left(F^{-1}(t)-t\right)\,dt,
\end{align*}
which satisfies $H'(\alpha)=0\Leftrightarrow \alpha=\int_0^1 \left(F^{-1}(t)-t\right)\,dt$.
\end{proof}

\begin{proof}[Proof of Lemma \ref{lemma:free_null_distribution}] Let $\mathcal{D}([0,1])$ denote the Banach space of right-continuous functions on $[0,1]$ with left limits. 
Donsker's Theorem \cite[Theorem 14.3]{billing}, states the weak convergence in $\mathcal{D}([0,1])$ of the empirical process $\sqrt{n}(F_n-F)$ for $n\rightarrow\infty$ to the standard Brownian bridge $\mathbb{B}(F(t))$. As the operator $h\,:\,\mathcal{D}([0,1])\longrightarrow\mathbb{R}$ defined as
\begin{equation}
    h(f)=\int_0^1\left(f(t)-\int_0^1 f(s)\,ds\right)^2\,dt=\int_0^1 f(t)^2\,dt-\left(\int_0^1 f(s)\,ds\right)^2
\end{equation}
is continuous, the continuous mapping Theorem \cite[Theorem 1.3.6]{vaart1996weak} yields that
\begin{equation}
    n\mathcal{T}_2(P^c_n,U)\underset{n}{\overset{w}{\longrightarrow}}\int_0^1 \mathbb{B}(t)^2\,dt-\left(\int_0^1\mathbb{B}(t)\,dt\right)^2,
\end{equation}
when $P^c=U$, which concludes the proof.
\end{proof}

\begin{proof}[Proof of Proposition \ref{prop:free_circle_statistic}] Keeping the notation of the proof of Lemma \ref{lemma:free_null_distribution}, after Theorem 1 in \cite{Ramdas2015} we have that, when $P^c=Q^c$,
\begin{equation}
    \sqrt{\frac{nm}{n+m}}\left(G_m^{-1}(F_n)-\mathbb{I}\right)\underset{n,m}{\overset{w}{\longrightarrow}}\mathbb{B}(t),
\end{equation}
in $\mathcal{D}([0,1])$, where $\mathbb{I}$ denotes the identity function. Finally, using the same arguments as in the proof of Lemma \ref{lemma:free_null_distribution}, the result is proved. 
\end{proof}

\begin{proof}[Proof of Proposition \ref{prop:consistency_marg}] Note that
\begin{equation*}
    \mathbb{P}(\pi_{nm}^c=1)=\mathbb{P}\left(T_{nm}^c\geq c_{nm}^c(\alpha)\right)=\mathbb{P}\left(\mathcal{T}_2(G_m\#P_n,U)\geq\frac{n+m}{nm}c^c_{nm}(\alpha)\right).
\end{equation*}
On one hand, we have that
\begin{align*}
    c_{nm}^c(\alpha)=\inf\lbrace t>0:F_{nm}^c(t)\geq 1-\alpha\rbrace=\inf\lbrace t>0:\mathbb{P}_{H_0}(T_{nm}^c>t)\leq \alpha\rbrace=\\
    \frac{nm}{n+m}\inf\lbrace t>0:\mathbb{P}_{H_0}(\mathcal{T}_2(G_m\#P_n,U)>t)\leq \alpha\rbrace.
\end{align*}
Under the null hypothesis, $G_m\#P_n\overset{w}{\rightarrow}U$. Thus, $\mathcal{T}_2(G_m\#P_n,U)\rightarrow 0$ in probability (recall Section \ref{section_asyimptotics}). In consequence, for every $\varepsilon>0$ and every $\alpha>0$, we have
\begin{equation}\label{one_hand}
    \frac{n+m}{nm}\,c_{nm}^c(\alpha)\leq \varepsilon
\end{equation}
for sufficiently large $n,m$. On the other hand, when $P^c\neq Q^c$, $\mathcal{T}_2(G_m\#P_n,U)\rightarrow \mathcal{T}_2(G\#P,U)>0$ in probability which, together with \eqref{one_hand}, proves the result.
\end{proof}

\begin{proof}[Size of \eqref{torus_marg_test}] \added{Let us prove that \eqref{torus_marg_test} controls the type I error at any significance level $\alpha>0$. Indeed, if $H_0$ denotes the null hypothesis \eqref{null_torus}, we have}
\begin{align}\label{sizeeq}\added{
   \mathbb{P}_{H_0}\left(\pi_{nm,N_g}^{g}=0\right)= \mathbb{P}_{H_0}\left(\min_{i=1}^{N_g} p_i \leq \frac{\alpha}{N_g}\right)=\mathbb{P}_{H_0}\left(\bigcup_{i=1}^{N_g} \left\lbrace p_i \leq \frac{\alpha}{N_g}\right\rbrace\right)\leq }\nonumber \\ \added{
   \sum_{i=1}^{N_g}\mathbb{P}_{H_0}\left(p_i\leq \frac{\alpha}{N_g}\right)= N_g\frac{\alpha}{N_g}=\alpha,}
\end{align}
\added{where the first equality in \eqref{sizeeq} is ensured as every $p_i$ follows a uniform distribution under the null.}
\end{proof}

\begin{proof}[Proof of Proposition \ref{prop:consistency_twomarg}]
\added{Suppose that $P^c_j\neq Q^c_j$ w.l.o.g. for some $j\in\lbrace 1,\ldots,N_g\rbrace$. After \eqref{torus_marg_test}, we have
\begin{align*}
    \mathbb{P}(\pi_{nm,N_g}^{g}=1)=\mathbb{P}\left(\min_{i=1}^{N_g} p_i\leq \frac{\alpha}{N_g}\right)\geq \mathbb{P}\left(p_j\leq \frac{\alpha}{N_g}\right).
\end{align*}
Then, as the right side of the previous inequality tends to $1$ after Proposition \ref{prop:consistency_marg}, so does the left side, which ends the proof.}
\end{proof}

\begin{proof}[Proof of Remark \ref{remark_abscont}]
\added{Suppose that $\mu_P\ll\ell_2$ and project with respect a given direction ${\bf e}_1$. As an immediate consequence of the monotone convergence theorem,  $\ell_{2}(A\times \R)=0$, for any Lebesgue null set $A\subset \R$. Consequently, by hypothesis $0=\mu_P(A\times \R)=\mu_P^1(A)$. Here, $\mu_P^1$ is the projected measure of  $\mu_P$ to the direction ${\bf e}_1$. Then, for any null set $B$ in $\mathbb{R}/\mathbb{Z}$ the leveraged set $ \tilde{B}=\bigcup_{s\in \Z}(s+B)$ is a Lebesgue null set, so 
that $\mu_P^1(\tilde{B})=0$ and  $P^c(B)=0$. }
\end{proof}

\begin{proof}[Proof of Theorem \ref{Theorem::bound}]
Note that $ \mathcal{T}_2(P_n,Q_m)=\mathcal{T}(X_1,\dots, X_n,Y_1,\dots, Y_m)$ is a function of $X_1,\dots, X_n$ and $Y_1,\dots, Y_m$.  For each $\x_1,\dots,\x_n,\y_1,\dots,\y_m\in\mathbb{T}^d$  and $\x'\in\mathbb{T}^d$ let $\pi$ and $\pi'$ be both joint measures such that
\begin{align*}
\mathcal{T}:=\sum_{i,j}&d(\x_i-\y_j )^2\pi_{i,j}=\mathcal{T}(\x_1,\dots, \x_n,\y_1,\dots, \y_m)\\
s.t. \ \ & \sum_{i,j}\pi_{i,j}=\frac{1}{n}, \  \ j=1,\dots, m,  \\
&\sum_{i,j}\pi_{i,j}=\frac{1}{m}, \ \ i=1,\dots, n,
\end{align*}

 and

\begin{align*}
\mathcal{T}':=\sum_{j}&d(\x'_1-\y_j )^2\pi'_{1,j}+\sum_{i>1,j}d(\x_i-\y_j )^2\pi'_{i,j}=\mathcal{T}(\x'_1,\dots, \x_n,\y_1,\dots, \y_m)\\
s.t. \ \ & \sum_{i,j}\pi'_{i,j}=\frac{1}{n}, \  \ j=1,\dots, m,  \\
&\sum_{i,j}\pi'_{i,j}=\frac{1}{m}, \ \ i=1,\dots, n.
\end{align*}
Then we have that
\begin{align*}
\mathcal{T}'\leq \sum_{j}d(\x'_1-\y_j )^2\pi_{1,j}+\sum_{i,j}d(\x_i-\y_j )^2\pi_{i,j},
\end{align*}
which implies
\begin{align*}
\mathcal{T}'-\mathcal{T}&\leq \sum_{j}\left( d(\x_i-\y_j )^2-d(\x_1-\y_j )^2\right)\pi_{1,j}\\
&\leq \sum_{j}\frac{1}{2}\pi_{1,j}=\frac{1}{2n},
\end{align*}
where the second inequality comes from the fact that  $d^2(\x,\y)\leq 1/2$ in $\mathbb{T}^d$.  By symmetry we also obtain the reverse inequality.  Doing the same with $\y'_1$ and $\y_1$ we obtain the bound $\frac{1}{2m}$. By using McDiarmid’s inequality, see \cite{mcdiarmid}, we derive that
\begin{align*}
\mathbb{P}\left(\mathcal{T}_2(P_n,Q_m)-\mathbb{E}\mathcal{T}_2(P_n,Q_m)>t\right)\leq \exp\left(-\frac{nm}{n+m} 8t^2 \right).
\end{align*}
\end{proof}

\begin{proof}[Proof of Proposition \ref{prop:epsilon}] Let $P=Q$. After Lemma \ref{Lemma:twosample_exp}, for every $\varepsilon>0$ these exists $N_\varepsilon\in\mathbb{N}$ such that for all $n,m\geq N_\varepsilon$, $\mathbb{E}\mathcal{T}_2(P_n,Q_m)\leq\varepsilon$. Using explicitly the convergence speed, we can find the relationship between $\varepsilon$ and $N_\varepsilon$:
\begin{equation}
    \frac{\log N_\varepsilon}{N_\varepsilon}=\frac{\varepsilon}{C},
\end{equation}
where $C>0$ is an unspecified constant. Then, directly from Theorem \ref{Theorem::bound}, we can bound \eqref{ideal_pvalue} as
\begin{gather*}
\mathbb{P}_{H_0}\left(\mathcal{T}_2(P_n,Q_m)>t\right)\leq \exp\left(-\frac{nm}{n+m} 8(t-\mathbb{E}\mathcal{T}_2(P_n,Q_m))^2 \right)\leq\\
\exp\left(-\frac{nm}{n+m} 8(t-\varepsilon)^2 \right),
\end{gather*}
for all $n,m\geq N_\varepsilon$.
\end{proof}

\begin{proof}[Proof of Proposition \ref{prop:asymptotic_consistency}]
\added{Let us first prove that \eqref{test_epsilon} is asymptotically of level $\alpha$. Let $\varepsilon>0$ and $N_\varepsilon\in\mathbb{N}$ such that for all $n,m\geq N_\varepsilon$, the test \eqref{test_epsilon} controls type I error. As we are taking the limit $n,m\rightarrow\infty$, we can choose $n,m$ large enough such that they surpass $N_\varepsilon$. Then, consistency is ensured by Proposition~\ref{prop:epsilon}.}

\added{To conclude, we prove the consistency under fixed alternatives such that $\mathcal{T}_2(P,Q)>\varepsilon$. First, note that
\begin{multline*}
     \mathbb{P}(\pi_{nm,\varepsilon}^{ub}=1)=\mathbb{P}\left(\mathcal{T}_2(P_n,Q_m)\geq\varepsilon +  \sqrt{-\frac{n+m}{8nm}\log\alpha}\right)\\=\mathbb{P}\left(\sqrt{\frac{mn}{n+m}}(\mathcal{T}_2(P_n,Q_m)-\mathcal{T}_2(P,Q))\geq \sqrt{\frac{mn}{n+m}}(\varepsilon-\mathcal{T}_2(P,Q)) +  \sqrt{-\frac{1}{8}\log\alpha}\right).
\end{multline*}
Now, \eqref{asymptotic_statistic_correctCent} implies, under the alternative, the stochastically boundedness of the left hand side.  However, the right hand side  is clearly unbounded if $ \mathcal{T}_2(P,Q)>\epsilon $. Consequently,
\begin{equation*}
    \underset{n,m\rightarrow\infty}{\lim} \mathbb{P}(\pi_{nm,\varepsilon}^{ub}=1)=1,
\end{equation*}
which concludes the proof.}
\end{proof}



\clearpage

\section{Supplementary figures}

\begin{figure}[h!]
\centering
\includegraphics[scale=0.4]{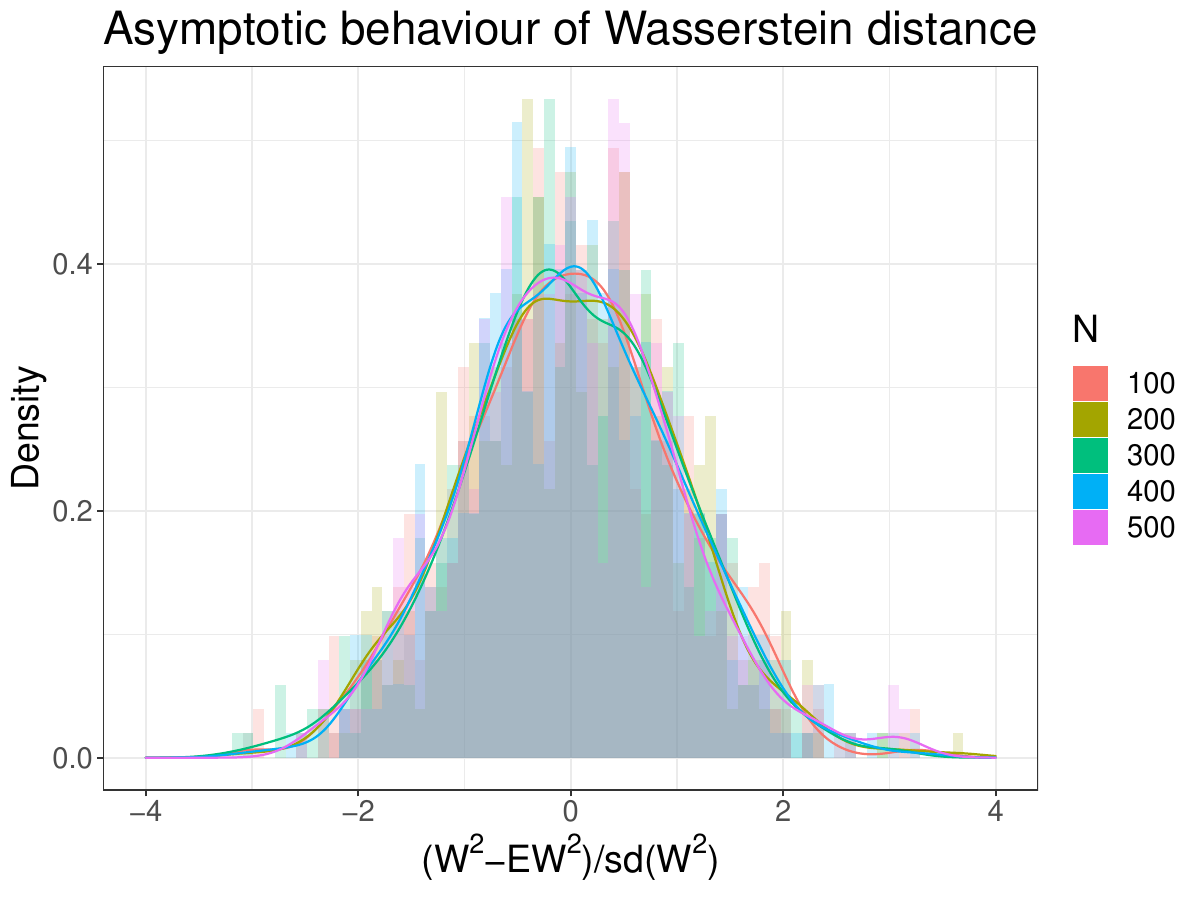}
\caption{Normalized asymptotic deviations from the mean of squared Wasserstein distance between two bivariate von Mises distributions of same means $(\mu,\nu)=(0,0)$ and different concentration parameters $(\kappa_1,\kappa_2,\kappa_3)=(0,0,0)$ and $(\kappa_1,\kappa_2,\kappa_3)=(2,2,0)$. The figures show the corresponding histograms and the associated kernel density estimates, for different sample sizes.}
\label{fig:tcl}
\end{figure}

\begin{figure}[h!]
\centering
\includegraphics[scale=0.4]{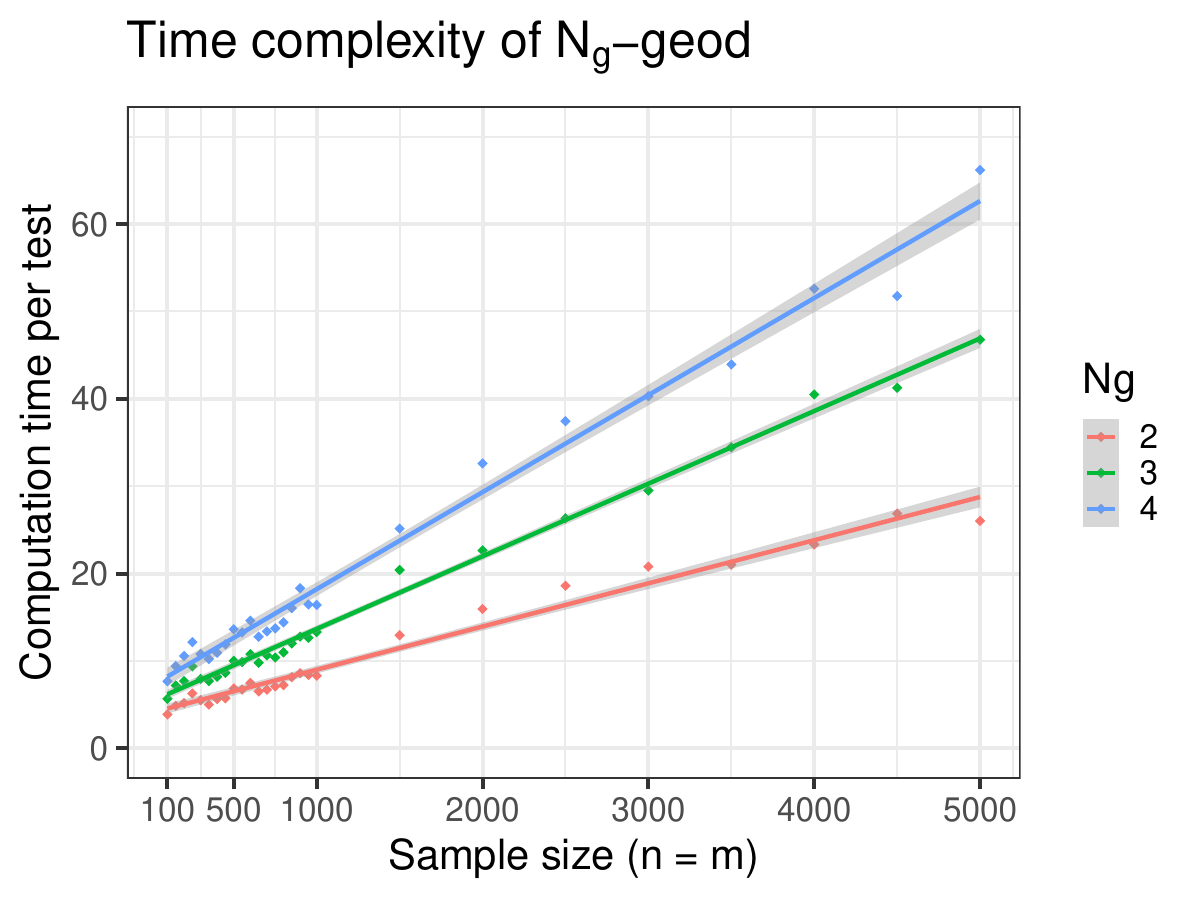}
\caption{Empirical time complexity of \eqref{torus_marg_test} for $N_g=2,3,4$. Each point corresponds to the average computation time per test among $200$ repetitions of \eqref{torus_marg_test} for two equally sized samples drawn from a bivariate von Mises distributions of equal means $(\mu,\nu)=(0,0)$ and different concentration parameters $(\kappa_1,\kappa_2,\kappa_3)=(0,0,0)$ and $(\kappa_1,\kappa_2,\kappa_3)=(1,1,0)$. The lines correspond to a linear regression performed for each value of $N_g$.}
\label{fig:time}
\end{figure}

\bibliographystyle{ejs-template} 
\bibliography{Biblio_lip}       


\end{document}